\documentclass[runningheads]{llncs}
\newif\ifdraft
\drafttrue
\usepackage{graphicx}
\usepackage{underscore}       
\usepackage{tablefootnote}
\usepackage{algorithmic}
\usepackage[utf8]{inputenc}    
\usepackage[T1]{fontenc}       
\usepackage{amssymb}
\usepackage{multirow}
\usepackage{pgf}
\usepackage{tikz}
\usetikzlibrary{arrows,automata}
\usepackage{stmaryrd}
\usepackage{mathtools}
\usepackage{appendix}

\usepackage{hyperref}
\hypersetup{
	colorlinks,
	citecolor=brown,
	final,
	linkcolor=blue,
}
\newcommand\calF{{\cal F}}
\newcommand\calG{{\cal G}}

\newcommand\sloc[1]{l^{\mathrm s}_{#1}}
\newcommand\dloc[1]{l^{\mathrm d}_{#1}}
\usepackage[normalem]{ulem}
\usepackage{thmtools,thm-restate}
\usepackage{enumerate}
\spnewtheorem{notation}{Notation}{\bfseries}{\rmfamily}

\newcommand\nstep[2]{#1^{[#2]}}
\begin{document}
	\title{Polynomial Probabilistic Invariants and the Optional Stopping Theorem}
	\author{Anne Schreuder\inst{1} \and
		C.-H.~Luke Ong\inst{2}}
	%
	\authorrunning{A.~Schreuder, C.-H.~L.~Ong}
	%
	\institute{Rheinische 	Friedrich-Wilhelms-Universität Bonn
		\and
		University of Oxford}
	
	\maketitle              
	\begin{abstract}
		In this paper we present methods for the synthesis of polynomial invariants for 
		probabilistic transition systems. 
		Our approach is based on martingale theory.
		We construct invariants in the form of polynomials over program variables, which give rise to martingales.
		These polynomials are program invariants in the sense that their expected value upon termination is the same as their value at the start of the computation. 
		In order to guarantee this we apply the Optional Stopping Theorem.
		Concretely, we present two approaches. The first is restricted to linear systems. In this case under positive almost sure termination there is a reduction to finding linear invariants for deterministic transition systems.
		Secondly, by exploiting geometric persistence properties we construct martingale invariants for general polynomial transition system.
		We have implemented this approach and it works on our examples.
	\end{abstract}
	%
	%
	%
		
	\section{Introduction}
	
	Probabilistic programs are computer programs that make random choices. 
	Their increasingly widespread use in many areas of computer science --- ranging across machine learning, network protocols and robotics --- has led to a recent surge of interest in the semantics of probabilistic programs and methods for reasoning about them.
	In this paper we present algorithmic solutions to a central verification problem: automatic extraction of program invariants.
	
	Finding invariants of probabilistic programs is hard.
	Given an input there are numerous different ways in which a probabilistic program can behave.
	A useful invariant must 
	take into account, simultaneously, all eventualities and the effects of randomness.
	
	Following \cite{forsyte,g,DBLP:conf/atva/FengZJZX17,DBLP:conf/cav/ChenHWZ15,DBLP:conf/qest/GretzKM13,DBLP:conf/sas/KatoenMMM10,b}, we explore methods that can automatically synthesize expressions over the program variables whose expectation is invariant over the course of the computation. 
	Mathematically the invariants we construct are \emph{martingales}. 
	Elegant and effective, martingales are a popular method in the static analysis of probabilistic programs 
	\cite{a,c,d,DBLP:journals/corr/abs-1904-01117,DBLP:conf/tacas/KuraUH19,g,h,DBLP:journals/corr/abs-1901-06087}.
	
	In this work, we restrict ourselves to invariants of probabilistic programs that still hold upon termination. 
	For non-probabilistic programs this distinction does not exist. 
	However, the run-time of a probabilistic program is generally not a fixed number but a probability distribution.
	Consequently the expectation of an invariant at the point of termination (\emph{qua} stopping time) may differ from its expectation after any fixed number of steps. 
	First observed in betting strategies, this phenomenon has been analysed extensively in probability theory, leading to Doob's celebrated Optional Stopping Theorem. 
	The Theorem is a collection of preconditions for the expectation of a martingale 
	taken at a stopping time to coincide with the expectation taken at the beginning.
	%
	
	The Optional Stopping Theorem (OST) has recently been successfully applied to the static analysis of probabilistic programs \cite{DBLP:journals/corr/abs-1904-01117,g,DBLP:journals/corr/abs-1901-06087,c}.
	However the phenomenon of the ``extra step'' required to check suitable preconditions had already been observed in the analysis of of nested loops.
	When reasoning about such probabilistic programs, it is common to use inductive arguments for the inner loop. 
	For reasons similar to those described earlier, it is necessary to impose appropriate OST-type preconditions on the inner-loop invariants, in order to formulate sound proof rules. 
	These methods have successfully been applied in \cite{wp1,DBLP:conf/atva/FengZJZX17,DBLP:conf/cav/ChenHWZ15,DBLP:conf/qest/GretzKM13}.
	As eloquently argued by Huang et al.~\cite{DBLP:journals/corr/abs-1901-06087}, great care is needed when reasoning about probabilistic programs.

We work with \emph{probabilistic transitions system} (PTS) \cite{a,b,h}, a general model of probabilistic imperative programs. 
Any standard imperative program with random samples and (nested) loops can be expressed as a PTS.

We focus on two preconditions of the OST. 
The first, called (PDB), requires the program to be {\bf p}ositively almost-sure terminating and the martingale expressions to be {\bf d}ifference {\bf b}ounded. 
The other, which we call (IUD), requires the existence of an {\bf i}ntegrable function that {\bf u}niformly {\bf d}ominates the martingale expressions. 

To the best of our knowledge, the (IUD) precondition is new. 
A key advantage of (IUD) is that it does not directly impose any restrictions on the stopping time (run-time), nor does it require the program variables to be bounded.
In the literature, such OST preconditions on the martingales (i.e.~with no recourse to stopping time) are either presented as uniform integrability \cite{c} which is difficult to check, or bounded by a constant \cite{DBLP:journals/corr/abs-1904-01117,DBLP:conf/atva/FengZJZX17,DBLP:conf/cav/ChenHWZ15,DBLP:conf/qest/GretzKM13} which is a special case of the (IUD) precondition.
From each of these preconditions we have derived a corresponding method for synthesizing invariants.

As a first contribution, we show that for linear programs that are positively almost-sure terminating, the search for invariants can be reduced to the case of deterministic (non-probabilistic) programs. 
Given a linear probabilistic program, we construct its determinised counterpart by taking the expectations of all the random samples. 
The price to pay for this reduction is that the invariants for the corresponding deterministic program need to be linear and difference bounded so that they can be carried over to the probabilistic case.
The advantage is that any linear invariant (expressible as a martingale) satisfying the (PDB) precondition can be found in this manner.
We are confident that this reduction will prove to be fruitful in future research as both the generation of invariants of deterministic programs, and reasoning about positive almost sure termination of linear probabilistic programs \cite{e,DBLP:journals/corr/abs-1901-06087,DBLP:conf/tacas/KuraUH19,d,a,c} have been studied in detail.


Our second contribution is a method to automatically synthesize martingale expressions of polynomial probabilistic transition system, via our new (IUD) precondition of the OST.
This is an extension of the approach by Barthe et al.~\cite{g}. 
Based on a ``seed'' polynomial $P$ they apply Doob's Decomposition Theorem in order to derive a martingale. 
They then ascertain whether this martingale satisfies the (PDB) precondition of the OST.

Our main improvement of this method exploits persistence properties of PTSs in the sense of Chakarov et al.~\cite{h}. 
We show that every PTS $\Pi$ gives rise to a ``persistence'' decision problem, which asks if there exists a sum-of-squares (SOS) polynomial $V$ of the program variables $x$ such that the expected value of $V$ at the next computation step of $\Pi$ is always no greater than some fixed fraction of the value of $V$, for all values of $x$.
Crucially, solutions to this decision problem are witnesses of the (IUD) precondition of the OST (in the sense that each summand $P$ of every SOS polynomial solution $V$ to the persistence decision problem determines an integrable function which uniformly dominates the martingale derived from $P$ via Doob's Decomposition).
We further reduce the persistence decision problem to a (standard) SOS Optimisation Problem, so that if the latter (feasibility) problem is solvable then the (IUD) precondition holds, thus allowing us to infer via OST that the martingale derived from $P$ is an invariant of the desired kind. 
The time complexity of this synthesis algorithm is polynomial in the size of the input system $\Pi$, when restricted to solutions $V$ of bounded degrees.

We have implemented this method which works on the examples in this paper, 
non-linear invariants, and non-linear systems.


\paragraph{Outline.}

In Sec.~\ref{sec: background}, we review standard background in martingale theory and fix notations.
In Sec.~\ref{sec: program}, we introduce probabilistic transition systems -- our model of probabilistic imperative programs, and invariants as martingales.
In Sec.~\ref{sec: linear}, we show how the (PDB) precondition of the OST can be exploited to reduce invariants of linear programs to those of their determinised counterparts.
In Sec.~\ref{sec: geometric persistence}, we present our algorithm that successively transforms a PTS to an instance of the SOS Optimisation Problem whose solvability implies satisfaction of the (IUD) precondition.
In Sec.~\ref{sec:examples} we briefly discuss our tool implementation and two examples.
Finally we discuss related work in Sec.~\ref{sec: relatedwork}.

Proofs missing from the paper are presented in \cite[Appendix]{SchreuderO19}.
	
	
	\section{Martingales, optional stopping, and SOS optimisation} \label{sec: background}
	
	We briefly review standard notions in probability theory (see e.g.~\cite{grimmett,Williams91,durrett}) and fix notations along the way. 
	
	Recall that a function $X: \Omega \rightarrow \mathbb{R}$ defined on a probability space $(\Omega, \calF, \mathbb{P})$ is a {\em random variable} if $X$ is $\mathbb{B}(\mathbb{R})$-$\calF$-measurable, i.e.~for all $A \in \mathbb{B}(\mathbb{R})$ -- the $\sigma$-algebra generated by the open subsets of $X$, we have $X^{-1}(A) \in \calF$.	
	The $k$-$th$ \emph{moment} of $X$ is the integral $\int_{\Omega} X^k \ d\mathbb{P}$,
	which exists if $\int_{\Omega} |X|^k \ d\mathbb{P}$ is finite.
	%
	
	Henceforth we fix a probability space $(\Omega, \calF, \mathbb{P})$. 
	Filtration is a concept introduced to represent information known over time. 
	Formally, a \emph{filtration} is a finite or infinite sequence $\{\calF_n\}_{n \in \mathbb{N}}$ of $\sigma$-algebras over $\Omega$ s.t.~$\forall n \in \mathbb{N} \, . \, \calF_n \subseteq \calF_{n+1} \subseteq \calF$. 
	(To save writing, we will often elide the subscript ``$n \in \mathbb{N}$'' and write the sequence as $\{\calF_n\}$.)
	The $\sigma$-algebra $\calF_n$ represent the information that is known at time $n$. 
	A sequence of random variables $\{X_n\}$ is \emph{adapted} to $\{\calF_n\}$ just if $X_n$ is $\mathbb{B}(\mathbb{R})$-$\calF_n$-measurable for all $n$.
	A random variable $T: \Omega \rightarrow (\mathbb{N} \cup \{\infty\})$ is called a \emph{stopping time} with respect to a filtration $\{\calF_n\}$ if 
	$\{T \leq n\}\in \calF_n$ for all $n \in \mathbb{N}$.
	%
	
	Let $X$ be an integrable random variable (i.e.~$\int |X| d \mathbb{P} < \infty$), and $\calG$ be a $\sigma$-subalgebra of $\calF$. 
	The \emph{conditional expectation} $\mathbb{E}[X \mid \calG]$ is a $\calG$-measurable function s.t.~for all $A \in \calG$,
	\(
	\int_{A} X \ d \mathbb{P} = \int_{A} \mathbb{E}[X \mid \calG] \ d \mathbb{P}
	\).
	The conditional expectation is $\mathbb{P}$-a.s.~unique, linear and monotone in $X$. 
	
	Given a filtration $\{\calF_n\}$, an integrable, $\{\calF_n\}$-adapted sequence of random variables $\{X_n\}$ is a
	\begin{enumerate}[(i)]
		\item  {\em Martingale} if for all $n \in \mathbb{N}$,
		$\mathbb{E}[X_{n+1} \mid \calF_n] = X_n$ a.s.
		\item 	{\em Supermartingale} if for all $n \in \mathbb{N}$,
		$\mathbb{E}[X_{n+1} \mid \calF_n] \leq X_n$ a.s.
		\item 	{\em Submartingale}	if for all $n \in \mathbb{N}$,
		$\mathbb{E}[X_{n+1} \mid \calF_n] \geq X_n$ a.s.
	\end{enumerate}
	
	
	\begin{theorem}[Doob's Decomposition]
		{\rm \cite[Theorem 10.1]{klenke}}
		\label{thm: Doob's decomposition}
		Let $(\Omega, \calF, \mathbb{P})$ be a probability space. 
		If $E = \{E_n\}_{n \in \mathbb{N}}$ is a sequence of integrable random variables adapted to a filtration $\{\calF_n\}_{n \in \mathbb{N}}$ where each $E_n$ has a finite expected value, then $M = \{M_n\}_{n \in \mathbb{N}}$ is a martingale where
		\[
		M_n := 
		\left\{
		\begin{array}{ll} 
		E_0 & \hbox{if $ n = 0$}\\
		E_0 + \sum_{i = 1}^{n} (E_i - \mathbb{E}[E_i \mid \calF_{i-1}]) & \hbox{otherwise}
		\end{array} \right.
		\]
		If $E$ is already a martingale, then $M = E$.
	\end{theorem}

	\begin{example}
		Let $\{X_n\}$ be i.i.d.~random variables with $\mathbb{P}(X_n=1) = \mathbb{P}(X_n= -1) = 1/2$. 
		Then $M_n := \sum_{i=1}^n X_i$ is the position of a one-dimensional symmetric random walk, starting from the origin,
		and $\{M_n\}$ is a martingale adapted to $\{\sigma(X_1, \cdots, X_n)\}$ with $\mathbb{E}(M_n) = 0$ for all $n$.
		Now $T := \min \{n \mid M_n = 1\}$ is a stopping time; clearly $M_T = 1$. 
		Moreover, it is well-known that $\mathbb{P}(T < \infty) = 1$.
		Yet we have $\mathbb{E}(M_T) = 1 \not= 0 = \mathbb{E}(M_0)$. 
	\end{example}
	
	The issue is that $T$ is too large: $\mathbb{E}(T) = \infty$.
	It turns out that if we impose suitable boundedness conditions, then we will have $\mathbb{E}(M_T) = \mathbb{E}(M_0)$: 
	this is the classical Optional Stopping Theorem due to Doob.
	
	\begin{theorem}[Optional Stopping] 
		\label{thm: OST}
		Let $(\Omega, \calF, \mathbb{P})$ be a probability space, and $M = \{M_n\}_{n \in \mathbb{N}}$ be a (super)martingale adapted to a filtration $\{\calF_n\}_{n \in \mathbb{N}}$, and $T$ a stopping time.
		If any of the following preconditions holds
		\begin{description}
			\item[(PDB)] \emph{Positive almost-sure termination and Difference Bounded:} $\mathbb{E}(T) < \infty$ and $\exists L > 0 \, . \, \forall n \in \mathbb{N} \, . \, \mathbb{E}[| M_{n+1} - M_n | \mid \calF_n] \leq L$. \emph{\cite[Thm.~10.10]{Williams91}}
			\item[(IUD)] \emph{Integrable Uniformly Dominating Function:} there exists an integrable function $g : \Omega \rightarrow \mathbb{R}$ satisfying $\forall n \in \mathbb{N} \, . \, |M_n| \leq g$ almost surely.
		\end{description}
		then $M_{T}$ is integrable. Furthermore $\mathbb{E}(M_{T}) = \mathbb{E}(M_0)$ if $M$ is a martingale; and $\mathbb{E}(M_{T}) \leq \mathbb{E}(M_0)$ if $M$ is a supermartingale.
	\end{theorem}
	
	\begin{proof}
		(IUD): The stopped-at-$T$ process, $\{M_{n \wedge T}\}_{n \in \mathbb{N}}$, is a (super)martingale. 
		Since $g$ is integrable, the family $\{g\} \subseteq L^1$ is uniformly integrable.
		It follows (from \cite[Theorem 6.18(iii)]{klenke}) that $\{M_{n \wedge T}\}_{n \in \mathbb{N}}$ is a uniformly integrable (super)martingale.
		By \cite[Theorem 11.7]{klenke} there is an integrable random variable $M_\infty$ with $\lim_{n \to \infty} M_{n \wedge T} = M_\infty$ a.s.~and in $L^1$. 
		Since $M_{n \wedge T}$ converges 
		to $M_{T}$ in distribution, we have $M_T = M_\infty$. 
		It then follows from convergence in $L^1$ that
		\(
		\lim_{n \rightarrow \infty} \mathbb{E}(M_{n \wedge T}) = \mathbb{E}(M_{T}).
		\)	
		Thus, by \cite[Theorem~p.~99]{Williams91},
		in case $\{M_{n \wedge T}\}_{n \in \mathbb{N}}$ is a martingale: for all $n \in \mathbb{N}$, 
		$\mathbb{E}(M_{0}) = \mathbb{E}(M_{n \wedge T})$,
		and so $\mathbb{E}(M_{0}) = \mathbb{E}(M_{T})$; in case $\{M_{n \wedge T}\}_{n \in \mathbb{N}}$ is a supermartingale: for all $n \in \mathbb{N}$, 
		$\mathbb{E}(M_{0}) \leq \mathbb{E}(M_{n \wedge T})$.
		Limits preserve inequalities, so
		\(
		\mathbb{E}(M_{0})
		= \lim_{n \rightarrow \infty} \mathbb{E}(M_{0})
		\leq
		\lim_{n \rightarrow \infty} \mathbb{E}(M_{n \wedge T}) = \mathbb{E}(M_{T}).
		\)	
		\qed
	\end{proof}

	There are numerous versions of the Optional Stopping Theorem (OST) in the literature, which differ in the precondition that is imposed (e.g.~\cite[\S 10.10]{Williams91}).
	Many, such as (PDB), require the stopping time to be bounded in some sense.
	One of very few exceptions is our precondition (IUD) which is new (to the best of our knowledge).
	
	In this paper, we apply Theorem~\ref{thm: OST} to reason about invariants synthesized from probabilistic programs.
	In our application, the stopping time $T$ will denote the run-time of a program. 
	Thus the formula $\mathbb{E}(T) < \infty$ in precondition (PDB) means that the expected run-time is finite.
	The random variable $M_n$ in Theorem~\ref{thm: OST} will be an expression built up from the program variables and locations in the $n$-th step of the computation. 
	Thus, the precondition (PDB) states, additionally, that the step-wise change of the value of the expresions $\{M_n\}$ is bounded;
	the precondition (IUD) requires that the value of $M_n$ is bounded above by an integrable function, uniformly (i.e.~for all computation steps $n$).
	
	\paragraph{SOS optimisation.}
	\label{subsec: sos opt}	
	
	Recall that a polynomial $p(x) \in \mathbb{R}[x]$ is a \emph{sum of squares} {\em (SOS)} if there exist $q_1, q_2, \cdots, q_m \in \mathbb{R}[x]$ such that 
	\(
	p(x) = \sum_{i = 1}^m q_i(x)^2.
	\)
	The {Sum-of-Squares (SOS) Optimisation Problem} has a linear cost function with a particular kind of constraint on the decision variables, namely, when the decision variables are used as coefficients in certain given polynomials, the resultant polynomials must be SOS. 
	Formally, the {\em SOS Optimisation Problem} asks: given $b \in \mathbb{R}^m$ and $a_{i,j}, c_i \in \mathbb{R}[x]$ for $1 \leq i \leq k$ and $1 \leq j \leq m$, compute
	\[
	\max_{y \in \mathbb{R}^m} \; \big(b_1 \cdot y_1 + b_2 \cdot y_2 + \cdots + b_m \cdot y_m\big)
	\]
	subject to $c_i(x) + a_{i,1}(x) \cdot y_1 +  \cdots + a_{i,m}(x) \cdot y_m$ is SOS (in $\mathbb{R}[x]$) for all $1 \leq i \leq k$.
	
	
	
	Any polynomial of degree at most $2d$ can be expressed in \emph{Gram matrix form},
	\(
	p(x) = z(x)^{T} \, Q \; z(x),
	\)
	where the (column) vector $z$ contains all monomials of degree $d$.
	In particular, a polynomial $p$ is SOS if and only if there exists a symmetric and positive-semidefinite matrix $Q$ such that 
	\(
	p(x) = z(x)^{T} \, Q \; z(x)
	\). 
	As a result the 
	SOS Optimisation Problem is equivalent to Semidefinite Programming (SDP) Problem \cite[page 74]{semidefiniteopt};
	and solutions to the SDP constraint systems can be found in polynomial time \cite[Sec.~2.3]{semidefiniteopt}.
	
	
	\section{Probabilistic transition systems and invariants} \label{sec: program}
	
	We consider probabilistic transition systems as models of imperative programs. 
	We distinguish program variables 
	${X} = \{x_1, x_2, \cdots, x_n\}$
	which are assigned deterministically, and random variables ${R} = \{r_1, r_2, \cdots, r_m\}$ which are sampled from probability distributions. 
	Both take values in $\mathbb{R}$.
	We assume that the random variables are
	stochastically independent of each other;
	moreover their distribution are i.i.d.~over time, and all moments exist.
	
	\begin{definition}[Probabilistic Transition Systems - PTS]\rm \cite{a,b,h}
		\label{def: pts}
		A \emph{probabilistic transition system} 
		is a tuple
		$\Pi = \langle W, X, R, X_0, L, l_0, l_F, T \rangle$ where
		\begin{enumerate}[(i)]
			\item
			$W = (\Omega, \calF, \mathbb{P})$ is a probability space.
			\item
			$X = \{x_1, x_2, \cdots, x_n\}$ are the program variables 
			and 
			$R = \{r_1, r_2, \cdots, r_m\}$ are the random variables defined on the probability space $W$.
			\item
			$X_0$ is the initial distribution of the program variables, where all moments exist.
			\item $L$ is a finite set of (program) \emph{locations};
			$l_0 \in L$ is the initial location and $l_F \in L$ is the final location.
			\item 
			$T = \{\tau_1, \tau_2, \cdots, \tau_p\}$ is a finite set of transitions.
			Each $\tau \in T$ is a quadruple $\langle \sloc{\tau}, \phi_{\tau}, f_{\tau}, \dloc{\tau} \rangle$ where $\sloc{\tau}$ is the source location and $\dloc{\tau}$ the destination location; $\phi_{\tau}$ is the {\em guard} \textit{assertion}, which is a boolean formula over polynomial inequalities over $X$; 
			and $f_{\tau}$, the {\em update function}, takes $(x, r) \in \mathbb{R}^{n} \times \mathbb{R}^{m}$ and returns (a discrete distribution 
			given by) the probability mass function
			\(
			\{F_{\tau, j}(x,r) \mapsto p_{\tau, j} \mid j \in \{1, 2, \cdots, j_\tau\}\}
			\)
			where 
			$F_{\tau,j} : \mathbb{R}^{n} \times \mathbb{R}^{m} \rightarrow \mathbb{R}^{n}$ are polynomials.
		\end{enumerate}
	\end{definition}

	A \emph{configuration} of a PTS $\Pi$ is a pair $(l, a)$ where $l \in L$, and $a \in \mathbb{R}^n$ represents the contents of the program variables.
	We say that a transition $\tau = \langle \sloc{\tau}, \phi_{\tau}, f_{\tau}, \dloc{\tau} \rangle$ is \textit{enabled} at a configuration $(l, a)$ if $l = \sloc{\tau}$ and $a \models \phi_{\tau}$.
	Let $l \in L$, we write $T_l := \{ \tau \in T \mid \sloc{\tau} = l\}$ for the set of transitions with $l$ as the source location.
	We assume that PTSs are \emph{non-demonic}:
	for all $l \in L$,
	$\bigwedge_{\tau \not= \tau' \in T_l}(\phi_{\tau} \wedge \phi_{\tau'}) = \mathit{false}$ and
	$\bigvee_{\tau \in T_l} \phi_\tau = \mathit{true}$.
	This amounts to \emph{determinacy} in the sense that for all $l$ and $a \in \mathbb{R}^n$, there is a unique transition that is enabled at $(l, a)$.
	
	\paragraph{PTS Computation.} The computation of a PTS $\Pi$ proceeds as follows.
	The initial configuration is $(l_0, a_0)$ where $a_0$ is sampled from $X_0$.
	Suppose $\Pi$ is at a configuration $(l, a)$ at step $n$. 
	Let $\tau$ be the unique transition that is enabled (i.e.~$l = \sloc{\tau}$ and $a \models \phi_\tau$), and $b \in \mathbb{R}^m$ be a fresh sample drawn from the respective distributions.
	Then, with probability $p_{\tau, i}$, the configuration at step $n+1$ is $(\dloc{\tau}, F_{\tau, i}(a, b))$.
	\begin{example} \label{ex: bsp}
		Consider the while loop	
		

		\medskip
		
		\noindent\begin{minipage}{0.5\textwidth}
			\begin{algorithmic}
				\WHILE {$x_1 \leq 10$} 
				\STATE {$r_1 \sim \textrm{Normal}(\mu_1, \sigma_1^2)$}
				\STATE {$r_2 \sim \textrm{Normal}(\mu_2, \sigma_2^2)$}
				\STATE {$x_1 := x_1 + r_1*x_1*x_2$}
				\STATE {$x_2 := x_2 + r_2*x_1*x_2$}
				\ENDWHILE
			\end{algorithmic}
		\end{minipage}%
		\begin{minipage}{0.5\textwidth}	
			where $x_1, x_2$ are program variables and $r_1, r_2$ are random variables. 
			Both $r_1$ and $r_2$ are normally distributed with mean $\mu_1$ and $\mu_2$, and variance $\sigma_1^2$ and $\sigma_2^2$, respectively.
		\end{minipage}	
		
		\medskip
		
		\noindent The while loop can be modelled by a PTS with transitions 
		$\tau_1 = \langle l_0, \phi_1, f, l_0 \rangle$ and
		$\tau_2 = \langle l_0, \phi_2, f, l_F \rangle$, 
		locations $\{l_0, l_F\}$,
		guards $\phi_1 = [x_1 \leq 10]$ and $\phi_2 = [x_1 > 10]$, and an update function $f : \mathbb{R}^2 \times \mathbb{R}^2 \to \mathbb{R}^2$ defined by
		\footnote{
			As we are only interested in the behaviour within the while loop, this is the simplest representation as a PTS.
		}
		\[
		f(x_1, x_2, r_1, r_2) = 
		(x_1 + r_1*x_1*x_2, \;
		x_1 + r_2*x_1*x_2)
		\]
		
	\end{example}
	
	
	%
	
	The operational semantics of a PTS $\Pi$ is given by
	random variables $\{ X^n, L^n \}$ on $W$ representing the distributions of the (contents of the) program variables and the location at each computation step $n$.
	
	\begin{definition}\rm 
		\label{def: operational semantics PTS}
		The {\em operational semantics} of a PTS $\Pi$ is a sequence of random variables $\{X^k, L^k\}_{k \in \mathbb{N}}$ where $X^k : \Omega \rightarrow \mathbb{R}^n$, $L^k: \Omega \rightarrow L$ satisfy
		\begin{enumerate}[(i)]
			\item
			$X^0$ 
			is the random variable $X_0$ as defined in Definition \ref{def: pts}
			\item  
			$L^0(\omega) = l_0$ is the initial location for all $\omega \in \Omega$
			\item 
			Let $\omega \in \Omega$, suppose
			$\tau$ is the unique transition enabled at the configuration $(X^k(\omega), L^k(\omega))$ then
			$L^{k+1}(\omega) = \dloc{\tau}$
			and 
			$X^{k+1}(\omega) = F_{\tau, i} \left(X^k(\omega), R^{k+1} \right)$
			for some $i \in \{1, \cdots, j_{\tau}\}$
		\end{enumerate}
		where $R^k$ denotes the vector of random samples drawn at the $k$-th step of the computation.

	\end{definition}
	
	We can show (see Lemma 2 in \cite{SchreuderO19}) that $\{X^k, L^k\}$ is adapted to the (natural) filtration $\{\calF_k\}$ where
	$\calF_k := \sigma(X^0, \widetilde{R}^1, \cdots, \widetilde{R}^k)$ where each $\widetilde{R}^i$ consists of the random variables $R^i$ from Definition~\ref{def: pts} and the discrete choice of the update function.
	Intuitively, $\calF_k$ contains information on the initial value $\left( X_0, l_0 \right)$ and the first $k$ random samples. 
	In other words, $\calF_{k}$ contains the information of the program run up to the $k$-th step of the computation.
	Thus  
	$\left\lbrace X^k, L^k \right\rbrace$ being adapted to ${\calF_k}$ states that $\left\lbrace X^k, L^k \right\rbrace$ is determined by the information denoted by $\calF_{k}$.
	
	As invariants we consider functions of the program variables. Their expected value of the next step can be expressed by the current values of the program variables.
	This concept was used in \cite{a,b,d,h,DBLP:conf/tacas/KuraUH19}.\footnote{
		Note that in \cite{b,h} there are no (program) locations defined. In this case there is one expression for the entire program.
	}
	
	\begin{definition}[Pre-expectation]\rm
		\label{def: pre-expectation}
		Let $\tau \in T$ and $h: \mathbb{R}^n \times L \times \mathbb{N} \rightarrow \mathbb{R}$ be a measurable function over valuations of the program variables, locations, and the step counter.
		Then the {\em pre-expectation} pre$\mathbb{E}(h, \tau) : \mathbb{R}^n \times \mathbb{N} \rightarrow \mathbb{R}$ is the function
		\[\textstyle
		(x, k) \mapsto 
		\mathbb{E}(h(f_{\tau}(x,R), \dloc{\tau}, k)) 
		= 
		\sum_{i = 1}^{j_{\tau}} p_{\tau,i} \, 
		\mathbb{E}_R(h(F_{\tau,i}(x, R), \dloc{\tau},k))
		\]
		and the pre-expectation\index{pre-expectation} for the PTS, pre$\mathbb{E}(h): \mathbb{R}^n \times L \times \mathbb{N} \rightarrow \mathbb{R}$, is given by 
		\[
		(x, l, k) \mapsto \sum_{i = 1}^{p} [(x \models \phi_i) \wedge (l = \sloc{\tau_i})] \cdot \mathrm{pre}\mathbb{E}(h,\tau_i)(x,k)
		\]
		where $\mathbb{E}_{R}$ is the expectation w.r.t.~the distribution of the random variables $R$ and $[{-}]$ is the Iverson bracket. 
	\end{definition}
	
	
	\subsection*{Probabilistic invariants}
	\label{sec:invariant}
	Mathematically, our probabilistic invariants are martingales, i.e.~sequences of random variables that are expectation invariant in each step.
	The history of the  computation is denoted by the  filtration $\{\calF_k\}_{k \in \mathbb{N}} = \{\sigma(X_0, \widetilde{R}^1, \widetilde{R}^2, \cdots, \widetilde{R}^k)\}_{k \in \mathbb{N}}$.
	
	These martingales can be constructed via the pre-expectation.
	
	
	\begin{restatable*}{lemma}{lemmatilde}
		\label{lemma tilde}
		Let $\Pi$ be a PTS with operational semantics $\{X^k, L^k\}_{k \in \mathbb{N}}$, and $h: \mathbb{R}^n \times L \times \mathbb{N} \rightarrow \mathbb{R}$ be a (measurable) function where $h(-, l, -) : \mathbb{R}^n \times \mathbb{N} \rightarrow \mathbb{R}$ is a polynomial $P_{l}$ for each $l \in L$.
		Then 
		\begin{enumerate}[(i)]
			\item $\mathbb{E}[h(X^{k+1}, L^{k+1}, k+1) \mid \calF_k] = \mathrm{pre}\mathbb{E}(h)(X^k, L^k, k+1)$ for all $k \in \mathbb{N}$.
			\item Further, if $\mathrm{pre}\mathbb{E}\left(h \right)(x, l, k+1) \leq h(x, l, k)$ for all $x \in \mathbb{R}^n, l \in L$ and $k \in \mathbb{N}$,		
			then $\left\{h(X^{k}, L^{k}, k)\right\}_{k \in \mathbb{N}}$ is a supermartingale adapted to $\{\calF_k\}_{k \in \mathbb{N}}$.
		\end{enumerate}
	\end{restatable*}

	\begin{remark}\label{rem: replacing powers} 
		\begin{enumerate}[(i)]	
			\item It follows that
			$\mathrm{pre}\mathbb{E}(h)(X^k)$ can be computed by replacing powers of random variables $r_i^a$ by their respective moments $\mathbb{E} \left( r_i^a \right)$ in the expression $h(X^k)$. 
			\item \label{rem: theorem martingale}
			By replacing $h$ by $-h$ the claim holds also for (sub)martingales.
			\item\label{rem: past}
			This lemma could prove positive a.s.~termination:
			Suppose $h(x,l,k) < 0$ iff~$l = l_F$, i.e.~the program has terminated.
			Suppose $\exists K < 0$, $\exists \varepsilon > 0$ s.t.~$\mathrm{pre}\mathbb{E}(h)(x, l, k+1) \leq h(x, l, k) - \varepsilon \cdot 1_{\{h(x, l, k) \geq 0\}}$ and $h(x, l, k) \geq K$.
			Then by \cite[Proposition 1]{d} the stopping time $T := \min\{n \in \mathbb{N} \mid h(X^n, L^n, n) < 0\}$ is a.s.~finite and
			$\mathbb{E}(T) \leq (\mathbb{E} (h(X^0, L^0, 0)) -K)/\varepsilon$.				
		\end{enumerate}
	\end{remark}
	Therefore, we construct martingales via the pre-expectation.
	If $\{M_k\}_{k \in \mathbb{N}}$ is a martingale, $\mathbb{E}(M_k) = \mathbb{E}(M_0)$ holds for all $k \in \mathbb{N}$  
	\cite[page 458]{borokov}.
	Hence $\mathbb{E}(M_k)$ is an invariant. 
	However, this does not imply $\mathbb{E}(M_{T}) = \mathbb{E}(M_0)$ where $T: \Omega \rightarrow  (\mathbb{N}  \cup \{\infty\}) $ is the stopping time representing the total number of computation steps executed.
	In particular, for a while loop such an invariant can be violated upon exiting the loop.

	\begin{notation}
		\label{notation}
		We write $\nstep{x_i}{j}$ for the \emph{random variable} describing the distribution of the program variable $x_i$ in the $j$th step of the PTS computation.
	\end{notation}
	
	\begin{example}[Martingale betting strategy]
		The gambler starts by wagering 1 unit on an evens bet. 
		So long as she loses, she continues by wagering twice the most recent bet on the next play.
		Thus when she eventually (and inevitably) wins, the amount won will cover all her previous loses and profit her by one unit.
		
		\medskip
		
		\noindent\begin{minipage}{0.4\textwidth}
			\centering
			\begin{algorithmic}
				\STATE {$x_1, x_2 := 1, 0$} 
				\WHILE {$x_2 \leq 0$}
				\STATE {$x_1 := 2 \cdot x_1$}
				\STATE {{$r_1 \sim \mathrm{uniform}(0,1)$}}
				\IF {$r_1 \leq 0.5$}
				\STATE {$x_2 := x_2 - x_1$}
				\ELSE
				\STATE {$x_2 := x_2 + x_1$}
				\ENDIF
				\ENDWHILE
			\end{algorithmic}
		\end{minipage}%
		\begin{minipage}{0.6\textwidth}	
			Consider the while loop
			where {$x_1$ (current bet) and $x_2$ (her stake thus far)} are the program variables and $r_1$ is the random variable.
			The expression
			$M_k := \nstep{x_2}{k}$ gives a martingale for this loop. Hence for all $k \in \mathbb{N}$, 
			$\mathbb{E}(M_k) = 0$
			$< 1 = \mathbb{E}(M_{T})$
			where $T := \inf\{k \in \mathbb{N} \mid \nstep{x_2}{k} > 0 \}$ is the exit time of the loop.
			
			\medskip
			
			Fortunately, the Optional Stopping Theorem guarantees $\mathbb{E}(M_{\tau}) = \mathbb{E}(M_0)$, 
			provided certain preconditions 
			are satisfied.
		\end{minipage}		
	\end{example}

	
	
	\section{Linear programs and the (PDB) precondition}
	\label{sec: linear}
	
	Linear probabilistic transition systems that are positively almost-sure terminating are  readily amenable to automated analysis:
	the search for invariants can be reduced to that of deterministic (non-probabilistic) PTS.
	
	
	\begin{definition}\rm 
		\label{def: corresponding pts}
		Let $\Pi = \langle W, X, R, X_0, L, l_0, l_F, T \rangle$ be a PTS.
		The {\em corresponding deterministic transition system (CDTS)} $\Pi^{\mathsf{det}}$ is a tuple
		$\langle X, X_0, L, l_0, l_F, T^{\mathsf{det}} \rangle$ where (using the notation of Def.~\ref{def: pts}) 
		$X = \{x_1, \cdots, x_n\}$ are the program variables from the PTS $\Pi$, and $T^{\mathsf{det}} = \{\tau^{\mathsf{det}}_1, \tau^{\mathsf{det}}_2, \cdots, \tau^{\mathsf{det}}_p\}$ is a finite set of transitions.
		For each transition $\tau = \langle \sloc{\tau}, \phi_{\tau}, f_{\tau}, \dloc{\tau} \rangle \in T$, there is a unique $\tau^{\mathsf{det}} = \langle \sloc{\tau}, \phi_{\tau}, f^{\mathsf{det}}_{\tau}, \dloc{\tau} \rangle \in T^{\mathsf{det}}$, where the deterministic update function, $f^{\mathsf{det}}_{\tau}: \mathbb{R}^{n} \rightarrow \mathbb{R}^{n}$, is defined by
		\(
		f^{\mathsf{det}}_{\tau}(x) := \sum_{i = 1}^{j_{\tau}} p_{\tau, i} \cdot F_{\tau, i}(x, \mu)
		\)
		with $\mu$ denoting the expectation of the random samples, and $F_{\tau, i}$ the constituent polynomials of the update function $f_\tau$ of $\Pi$.
		
	\end{definition}
	
	\begin{definition}\rm 
		\label{def: deterministic pre-expectation}
		Let $h: \mathbb{R}^n \times L \times \mathbb{N} \rightarrow \mathbb{R}$ be a measurable function, and let $\tau \in T^{\mathsf{det}}$ be a transition of $\Pi^{\mathsf{det}}$, the CDTS of $\Pi$.
		The {\em deterministic pre-expectation} pre$\mathbb{E}^{\mathsf{det}}(h, \tau) : \mathbb{R}^n \times \mathbb{N} \rightarrow \mathbb{R}$ is the function
		$(x, k) \mapsto h(f^{\mathsf{det}}_{\tau}(x), \dloc{\tau}, k)$.
		The \emph{pre-expectation\index{pre-expectation} for $\Pi^{\mathsf{det}}$} is 
		pre$\mathbb{E}^{\mathsf{det}}(h): \mathbb{R}^n \times L \times \mathbb{N}\rightarrow \mathbb{R}$ defined by
		\[\textstyle
		(x, l, k) \mapsto
		\sum_{i = 1}^{p} [(x \models \phi_i) \wedge (l = \sloc{\tau_i})] \cdot \mathrm{pre}\mathbb{E}^{\mathsf{det}}(h,\tau_i)(x, k).
		\]
	\end{definition}
	
	The expectation is linear, which gives us a correspondence between linear invariants of the probabilistic transistions and their determinised counterparts. 

	\begin{restatable*}{lemma}{linearOST}
		\label{lemma: linear ost non-appendix}
		Let $\Pi = \langle W, X, R, X_0, L, l_0, l_F, T \rangle$ be a linear PTS and $\tau \in T$.
		Given a function $h: \mathbb{R}^n \times L \times \mathbb{N} \rightarrow \mathbb{R}$ defined by $h(x, l, k) = \sum_{l'} \ [l' = l] \cdot P_{l'}(x, k)$ where $P_{l'}$ is linear for each $l' \in L$. Then 
		\begin{enumerate}[(i)]
			\item for all $k \in \mathbb{N}$ and $x \in \mathbb{R}^n$,
			\(
			\mathrm{pre}\mathbb{E}(h, \tau)(x, k)  
			= \mathrm{pre}\mathbb{E}^{\mathsf{det}}(h, \tau)(x, k)
			\)
			
			\item there is $K > 0$ s.t.~for all $k \in \mathbb{N}$, 
			$ 
			\mathbb{E}[| h(X^{k+1}, L^{k+1}, k+1) - h(X^{k}, L^k, k) | \mid \calF_{k}]
			\leq		\left|
			\mathrm{pre}\mathbb{E}^{\mathsf{det}}(h)(X^k, L^k, k+1)
			- h(X^{k}, L^k, k) \right| + K.
			$
		\end{enumerate}
	\end{restatable*}
	
	

	Lemma~\ref{lemma: linear ost non-appendix} enables use to prove the correctness of our reduction.
	
	\begin{theorem}[Correspondence]\label{thm:correspondence}
		Let $\Pi$
		be a linear PTS
		which is positively almost-sure terminating,
		and let $h: \mathbb{R}^n \times L \times \mathbb{N} \rightarrow \mathbb{R}$ be defined by $h(x, l, k) = \sum_{l' \in L} \ [l' = l] \cdot P_{l'}(x, k)$ for $P_{l'}$  linear.
		If either of the following conditions is satisfied,
		\begin{enumerate}[(i)]
			\item 
			$\mathrm{pre}\mathbb{E}^{\mathsf{det}}(h)(x, l, k+1) = h(x, l, k)$ for all $x \in \mathbb{R}^n$, $l \in L$,
			\item
			$\mathrm{pre}\mathbb{E}^{\mathsf{det}}(h)(x, l, k+1) \leq h(x, l,k )$
			and there exists $K > 0$ such that 
			for all $x \in \mathbb{R}^n$, $l \in L$,   $\left| \mathrm{pre}\mathbb{E}^{\mathsf{det}}(h)(x, l, k+1) - h(x, l, k) \right| \leq K$,
		\end{enumerate}
		then 
		$\{ h(X^k, L^k, k) \}_{k \in \mathbb{N}}$ is a (super)martingale for $\Pi$ satisfying the (PDB) precondition of the OST~\ref{thm: OST}.
	\end{theorem}
	
	\begin{proof}
		The (super)martingale property is a direct consequence of Lemma~\ref{lemma tilde}(ii).
		Thanks to Lemma \ref{lemma: linear ost non-appendix} and the positive almost-sure termination of $\Pi$, the (PDB) precondition of the OST \ref{thm: OST} is satisfied.
		\qed
	\end{proof}
	
	\begin{example} [Nested loop]
		\label{eg: nested loop}
		Consider the following nested-loop program \cite{d,DBLP:journals/corr/FengZJZX17} and its corresponding determinised program.
		
		\medskip
		
		\noindent\begin{minipage}{0.45\textwidth}
			
			\begin{algorithmic}
				\STATE {$x := 0$}
				\STATE {$z := 0$}
				\WHILE {$x \leq m$} 
				\STATE {$y := 0$}
				\WHILE {$y \leq n$} 
				\STATE {$\mathtt{r_1} \sim \mathrm{uniform}(-0.1, 0.2)$}
				\STATE {$y := y + \mathtt{r_1}$}
				\ENDWHILE
				\STATE {$\mathtt{r_2} \sim \mathrm{uniform}(-0.1, 0.2)$}
				\STATE {$x := x + \mathtt{r_2}$}
				\STATE {$z := z + 1$}
				\ENDWHILE
			\end{algorithmic}
		\end{minipage}
		\hfill
		\begin{minipage}{0.45\textwidth}
			\begin{algorithmic}
				\STATE {$x := 0$}
				\STATE {$z := 0$}
				\WHILE {$x \leq m$} 
				\STATE {$y := 0$}
				\WHILE {$y \leq n$} 
				\STATE {$y := y + 0.05$}
				\ENDWHILE
				\STATE {$x := x + 0.05$}
				\STATE {$z := z + 1$}
				\ENDWHILE
			\end{algorithmic}
		\end{minipage}
		
		\medskip
		
		\noindent
		Let $T: \Omega \rightarrow \mathbb{N} \cup \{\infty\}$ denote the total runtime of the program; 
		and for $k \in \mathbb{N}$, let $T_1^k: \Omega \rightarrow \mathbb{N} \cup \{\infty\}$ denote the run-time of the inner loop in the $k$-th run of the outer loop, and let $T_2: \Omega \rightarrow \mathbb{N} \cup \{\infty\}$ denote the number of times the outer loop body is executed.
		(Here we use the semantics of nested stopping times due to Fioriti and Hermanns \cite{c}.)
		
		Since the two loop bodies are independent, we have $\mathbb{E}(T_1^k) = \mathbb{E}(T_1^1)$ for all $k \in \mathbb{N}$.
		Moreover, $\mathbb{E}(T) = \mathbb{E}(T^1_1 \cdot T_2) = \mathbb{E}(T^1_1) \cdot \mathbb{E}(T_2)$.
		
		We first show that both loops terminate positively almost surely.
		For the inner loop consider $h(y) = n - y$.
		By Lemma \ref{lemma: linear ost non-appendix}, $\mathrm{pre}\mathbb{E}(h)(y) 
		= [y \leq n] \cdot h(y + 0.05) + [y > n] \cdot h(y) 
		= h(y) - 0.05 \cdot [h(y) \geq 0]$.
		Furthermore, as $\mathtt{r_1} \leq 0.2$, we have $h(y) \geq -0.2$ for all values of $y$ that arise within the inner loop.
		Moreover, $T_1^k = \min\{i \in \mathbb{N} \mid h(Y^i) < 0\}$ for all $k \in \mathbb{N}$.
		As a consequence by Remark \ref{rem: replacing powers}(\ref{rem: past}),
		$\mathbb{E}(T_1^k) \leq (n - (-0.2))/0.05= 20n+4$ for all $k \in \mathbb{N}$.
		Similarly, we can conclude 
		$\mathbb{E}(T_2) \leq 20m + 4$.
		Together we have $\mathbb{E}(T) \leq 400mn + 80(m+n)+16$.
		
		If we only consider
		the inner loop, 
		$\{\nstep{y}{k} - 0.05 k\}$ is clearly an invariant in both cases. So by Theorem \ref{thm:correspondence}, 
		$0 = \mathbb{E}(\nstep{y}{0} - 0.05 \cdot 0) = \mathbb{E}(\nstep{y}{T_1^1} - 0.05 \cdot T_1^1)$.
		Hence
		$n \leq \mathbb{E}(\nstep{y}{T_1^1}) = 0.05 \cdot \mathbb{E}(T_1^1)$.
		In particular, $\mathbb{E}(T_1^1) \geq 20n$.
		Similarly the computation of the variable $x$ in the outer loop does not depend on the inner loop.
		Ranging over the number of times the outer loop is executed $\{\nstep{x}{k} - 0.05 k\}$, $\{\nstep{z}{k} - k\}$ and $\{\nstep{z}{k} - 20 \nstep{x}{k}\}$ are invariants for both outer loops, by Theorem \ref{thm:correspondence}.
		Direct computation shows
		$\mathbb{E}(\nstep{x}{T_2}) = 0.05 \cdot \mathbb{E}(T_2)$, $\mathbb{E}(\nstep{z}{T_2}) = \mathbb{E}(T_2) = 20 \cdot \mathbb{E}(\nstep{x}{T_2}) \geq 20 m$.
		This gives a lower bound on the total run-time:
		$\mathbb{E}(T) = \mathbb{E}(T_1) \cdot \mathbb{E}(T_2) \geq 400mn$.
	\end{example}

	\section{Geometric persistence and the (IUD) precondition}
	\label{sec: geometric persistence}

	Let $\Pi$ be a PTS, and take a function $P: \mathbb{R}^n \times L \rightarrow \mathbb{R}$ such that $P(-, l)$ is a polynomial for each location $l \in L$. 
	Consider the random variable $P(X^k, L^k)$, where $X^k$ and $L^k$ are random variables giving the distributions of the program variables and location in the $k$-th step of the computation of $\Pi$.
	It follows that $\{P(X^k, L^k)\}_{k \in \mathbb{N}}$ is adapted to the (natural) filtration $\{\calF_k\}_{k \in \mathbb{N}}$; 
	moreover $\mathbb{E}(P(X^k, L^k)) < \infty$ because all moments exist by assumption.
	By Doob's Decomposition (Theorem~\ref{thm: Doob's decomposition}), we have that $\{M_k\}_{k \in \mathbb{N}}$, where
	\begin{equation}
	M_k := 
	\left\{
	\begin{array}{ll} 
	P(X^0, L^0) & \ \hbox{if $k = 0$}\\
	P(X^0, L^0) + \sum_{i = 1}^{k} \big(P(X^i, L^i) - \mathbb{E}[P(X^i, L^i) \mid \calF_{i-1}]\big) & \ \hbox{otherwise} 
	\end{array} \right.
	\label{eq:Doob}
	\end{equation}
	is a martingale.
	%
	%
	In this Section we present a method to synthesize such functions $P$ that the derived martingale $\{M_k\}_{k \in \mathbb{N}}$ is guaranteed to satisfy precondition (IUD) of the OST~\ref{thm: OST}.

	
	\subsection{A synthesis algorithm and its formulation as a SOS problem}
	\label{sec: algorithm}
	
	In order to meet the precondition (IUD), we need to construct an integrable function $g$ that uniformly dominates $|M_k|$.
	Inspired by Chakorav et al.~\cite{h}, we exploit persistence properties of the PTS (if they exist) to synthesize $g$, using sum-of-squares optimisation techniques.
	
	
	Our synthesis algorithm is simple. 
	Given a PTS $\Pi$, we solve the problem: 
	\begin{quote} Given $d \in \mathbb{N}$, construct function $V: \mathbb{R}^n \times L \rightarrow \mathbb{R}$ s.t.~for all $l \in L$:
		
		\begin{description}
			\item[(a1)] 
			$V(-, l)$ is a SOS polynomial of degree at most $d$, and
			
			\item[(a2)]
			$\exists \alpha \in [0,1) \; . \; \forall x \in \mathbb{R}^n \; . \;
			\mathrm{pre}\mathbb{E}(V)(x, l) \leq \alpha  \cdot V(x, l)$.
		\end{description} 
	\end{quote}
	If a solution $V$ is found, and for each $l \in L$, $P(-, l)$ is a summand in the SOS decomposition of $V(-, l)$ for all $l \in L$, then $P: \mathbb{R}^n \times L \rightarrow \mathbb{R}$ gives rise to a 
	martingale $\{M_k\}$ as defined in (\ref{eq:Doob}), satisfying
	$\mathbb{E}(P(X^0, l_0)) = \mathbb{E}(M_0) = \mathbb{E}(M_{\tau})$. 
	
	
	
	\medskip
	
	Recall Def.~\ref{def: pre-expectation},
	$
	\mathrm{pre}\mathbb{E}(V)(x, l)
	=
	\sum_{i = 1}^{p} [(x \models \phi_i) \wedge (l = \sloc{\tau_i})] \cdot \mathrm{pre}\mathbb{E}(V,\tau_i)(x)
	$
	where $p \in \mathbb{N}$ is the number of transitions, and the transition guards $\phi_i$ are conjunctions of polynomial inequalities.
	Condition (a2) is equivalent to
	\begin{equation}
	\exists \alpha \in [0, 1) \, . \, 
	\forall \tau \in T \, . \,
	\forall x \in \mathbb{R}^n\, . \, 
	\phi_{\tau}(x) \to \big(\underbrace{\alpha \cdot V (x, \sloc{\tau}) - \mathrm{pre}\mathbb{E}(V,\tau)(x)}_{Q(x)} \geq 0\big)
	\label{eq:SOS}
	\end{equation}
	
	By assumption $\phi_{\tau} = \bigwedge_{j=1}^{j_\tau} (p_j \geq 0)$ for some polynomials $p_1, \cdots, p_{k_\tau}$.
	The non-negativity of $Q(x)$ on a set constrained by polynomial inequalities (such as $\phi_\tau$) can be witnessed by suitable SOS polynomials \cite[\S 3.2.4, pp.~78--79]{semidefiniteopt}. 
	Thus (\ref{eq:SOS}) is implied by: 
	\[\begin{array}{l}
	\exists \alpha \in [0, 1) \; . \; 
	\forall \tau \in T \; . \; \exists q_0, \cdots, q_{k_\tau} \in \mathbf{SOS} \; . \; \forall x \in \mathbb{R}^n \; . \hfill \\
	\qquad \qquad \qquad \alpha \cdot V(x, \sloc{\tau}) - \mathrm{pre}\mathbb{E}(V,\tau)(x) = q_0(x) + \sum_{j = 1}^{k_\tau} {q_j(x)} \cdot p_j(x).
	\end{array}\]
	writing $\mathbf{SOS}$ for the set of SOS polynomials.
	
	As most SOS optimisation solvers do not handle strict inequalities, we choose some $\varepsilon > 0$ (e.g.~$10^{-6}$), and assume $\alpha \leq 1 - \varepsilon$.
	Thus we can encode the preceding condition by
	
	\begin{description}
		\item[(a2')] $\exists \epsilon > 0 \; . \; \forall \tau \in T \; . \; \exists
		q_0, \cdots, q_{k_\tau} \in \mathbf{SOS} \; . $\\ 
		\hspace*{2cm} $(1 - \epsilon) \cdot V(-, \sloc{\tau}) - \mathrm{pre}\mathbb{E}(V, {\tau}) - q_0 - \sum_{j = 1}^{k_\tau} {q_j} \cdot p_j \in \mathbf{SOS}$
	\end{description}
	
	%
	To compute
	$\mathrm{pre}\mathbb{E}(V,\tau)(x) = \sum_{i = 1}^{j_{\tau}} p_{\tau,i} \cdot \mathbb{E}_R(V(F_{\tau,i}(x,r), \dloc{\tau}))$,
	we only need to know the moments.
	The random samples are independent of the previous computation and of each other. The expected value $\mathbb{E}_R(V(F_{\tau,i}(x,r), \dloc{\tau}))$ is obtained by composing $V(-, \dloc{\tau})$ and $F_{\tau,i}$ 
	and replacing all occurrences of a power of a random variable $r_l^m$ by its moment $\mathbb{E}(r_l^m)$ (see Remark~\ref{rem: replacing powers}).
	
	\subsection{Correctness and complexity}
	
	Let $V: \mathbb{R}^n \times L \rightarrow \mathbb{R}$ be a non-constant function satisfying (a1) and (a2), 
	and let $P: \mathbb{R}^n \times L \rightarrow \mathbb{R}$ satisfy $V(x, l) \geq P^2(x, l)$ for all $x \in \mathbb{R}^n$ and $l \in L$.
	We prove that the martingale $\{M_k\}$ induced by $P$ (see (\ref{eq:Doob})) is a probabilistic invariant via the (IUD) precondition of the OST~\ref{thm: OST}.
	To achieve this, we construct an integrable function $g: \Omega \rightarrow \mathbb{R}$ satisfying $|M_k| \leq g$ for all $k \in \mathbb{N}$. 
	Now, it follows from Lemma \ref{lemma tilde} that
	\[
	\begin{array}{rll}
	M_k 
	&=& 
	P(X^0, L^0) + \sum_{i=1}^k \left(P(X^i, L^i) - \mathrm{pre}\mathbb{E}(P)(X^{i-1}, L^{i-1}) \right)
	\\ 
	&=& 
	\sum_{i=0}^k P(X^i, L^{i}) - \sum_{i=1}^k \mathrm{pre}\mathbb{E}(P)(X^{i-1}, L^{i-1})
	\\
	&=& 
	\sum_{i=0}^k P(X^i, L^{i}) - \sum_{i=0}^{k-1} \mathrm{pre}\mathbb{E}(P)(X^{i}, L^{i})
	\\ 
	&=& 
	\sum_{i=0}^{k-1} \left( P(X^i, L^i) - \mathrm{pre}\mathbb{E}(P)(X^{i}, L^i) \right) + P(X^k, L^k).
	\end{array}
	\]
	Thus we can set $g := \sum_{i=0}^{\infty} |P(X^i, L^i) - \mathrm{pre}\mathbb{E}(P)(X^{i}, L^i)| + \sum_{i=0}^{\infty} | P(X^k, L^k)|$. 
	Trivially we have $|M_k| \leq g$ for all $k \in \mathbb{N}$.
	To prove that $g$ is integrable, we establish the following:
	\begin{enumerate}[(i)]
		\item $\mathbb{E} \left( \sum_{i = 0}^{\infty} |P(X^i, L^i) - \mathrm{pre}\mathbb{E}(P)(X^i, L^i)| \right) < \infty$
		(\cite[Lemma~\ref{lem:exp-inf-sum-p-preEP}]{SchreuderO19})
		\item $\mathbb{E} \left( \sum_{i = 0}^{\infty} |P(X^i, L^i)| \right) < \infty$ (\cite[Lemma~\ref{lemma: |P| < infty}]{SchreuderO19})
	\end{enumerate}


	\paragraph{Complexity.}
	
	SOS polynomials are non-negative. 
	Of course, not every non-ne\-ga\-tive polynomial is a SOS polynomial 
	\cite[page 59]{semidefiniteopt}.	
	However, SOS optimisation problems can be solved in polynomial time, whereas deciding non-negativity of multivariate polynomials is NP-hard in general (see \cite[p.~66]{semidefiniteopt}).
	
	As discussed in Sec.~\ref{sec: background}, the SOS Optimisation Problem can be solved by reduction to the Semidefinite Programming problem. 
	The complexity of the reduction is polynomial in the number of new monomials. 
	For each location $l \in L$, we have monomials over $n$ program variables at degree up to $2 d$ where $d$ is here the maximal degree of our invariant generating polynomial $P$. This gives us $O(|L| \cdot n^{2 d})$ monomials.	
	By \cite[Sec.~2.3]{semidefiniteopt} the SDP program is also solved in polynomial time over $O(|L| \cdot n^{2 d})$.
	%
	Thus for $d \geq 0$ fixed, this approach has a polynomial time complexity.

	\section{Implementation and examples}
	\label{sec:examples}
	
	We have implemented the feasibility problem for $(a1)$, $(a2')$ using the SOS optimisation solver YALMIP \cite{yalmip2}, which is in turn based on the SDP optimisation solver SeDuMi \cite{sedumi}.
	Whenever the objective function is not specified, a feasible solution is computed and returned instead.
	The condition $(a1)$ can be quickly implemented using the command \textsf{sos}. 
	For this we used the function \textsf{polynomial} that allows us to create parametrized polynomials of chosen degree over given variables.
	The other condition $(a2')$ was implemented by the method outlined in Section~\ref{sec: algorithm}. 
	In order to be able to replace monomials of random variables by respective moments, we used the function \textsf{coefficients} that splits a polynomial for a given variable into monomials and their respective coefficients with respect to that variable.

	\subsection{Example: a probabilistic while loop}
	
	
	Here we consider the while loop from Example \ref{ex: bsp}.
	For several combinations of mean and variance we can find persistence properties of the form 
	$V(x_1, x_2) = (a x_1 + b x_2)^2 + (c x_1 + d x_2)^2$.
	For instance, if we choose $r_1 \sim \textrm{Normal}(-3, 1)$ and $r_2 \sim \textrm{Normal}(2, 2)$ we obtain
	(rounded)
	$V(x_1, x_2) = (0.74227x_1-0.086046x_2)^2 + (0.02481x_1+0.21398x_2)^2$.
	Thus summands of $V$ have the form $(a x_1 + b x_2)^2$ for fixed numbers $a, b$.
	Such a polynomial $P$ defines an invariant (using Notation~\ref{notation}):
	\[
	M_k := 
	\left\{
	\begin{array}{ll} 
	a \nstep{x_1}{0} + b \nstep{x_2}{0} 
	&
	\hbox{if $k = 0$}\\
	a \nstep{x_1}{0} + b \nstep{x_2}{0} + \sum_{i = 1}^{k} \big(a \nstep{x_1}{i} - b \nstep{x_2}{i} - \mathrm{pre}\mathbb{E}(P)(\nstep{x_1}{i-1}, \nstep{x_2}{i-1})\big) \quad 
	& \hbox{if $k > 0$} 
	\end{array} 
	\right.
	\]
	If we write $\mu_1 = \mathbb{E}(r_1)$ and $\mu_2 = \mathbb{E}(r_2)$ for the respective expected values, then a quick calculation yields 
	\begin{align*}	
	\mathrm{pre}\mathbb{E}(P)(x_1, x_2) 
	= & \ \mathbb{E}_{R} \left( a x_1 + a r_1 x_1 x_2 + b x_2 + b r_2 x_1 x_2 \right) 
	\\
	= & \ a x_1 + b x_2 + (a \mu_1 + b \mu_2) x_1 x_2
	\end{align*}
	So we can write $\left\lbrace M_k \right\rbrace_{k \in \mathbb{N}}$ as
	\[	
	M_k := 
	\left\{
	\begin{array}{ll} 
	a \nstep{x_1}{0} + b \nstep{x_2}{0} 
	&
	\hbox{if $k = 0$}\\
	a \nstep{x_1}{k} + b \nstep{x_2}{k} - \sum_{i = 0}^{k-1} (a \mu_1 + b \mu_2) \cdot \nstep{x_1}{i} \nstep{x_2}{i} \quad
	& \hbox{if $k > 0$} 
	\end{array} 
	\right.
	\]
	If we denote the loop condition by the stopping time $T = \min \{ n \in \mathbb{N} \mid x_1 > 10 \}$,
	then the following holds.
	\[\begin{array}{ll}
	a \mathbb{E}(\nstep{x_1}{0}) + b \mathbb{E}(\nstep{x_2}{0}) 
	&= \mathbb{E}(M_0) = \mathbb{E}(M_T) 
	\\
	&= a \mathbb{E}(\nstep{x_1}{T}) + b \mathbb{E}(\nstep{x_2}{T}) - \mathbb{E} \left( \sum_{i = 0}^{T-1} (a \mu_1 + b \mu_2) \cdot \nstep{x_1}{i} \nstep{x_2}{i}  \right)
	\end{array}\]
	Hence $a \mathbb{E}(\nstep{x_1}{T}) + b \mathbb{E}(\nstep{x_2}{T}) = a \mathbb{E}(\nstep{x_1}{0}) + b \mathbb{E}(\nstep{x_2}{0}) + \mathbb{E} \left( \sum_{i = 0}^{T-1} (a \mu_1 + b \mu_2) \cdot \nstep{x_1}{i} \nstep{x_2}{i}  \right)$.
	\\
	As we know that $\nstep{x_1}{T} > 10$ we conclude by the monotonicity of the expected value	
	\[\begin{array}{ll}
	b \mathbb{E}(\nstep{x_2}{T}) &= a \mathbb{E}(\nstep{x_1}{0}) + b \mathbb{E}(\nstep{x_2}{0}) + \mathbb{E} \left( \sum_{i = 0}^{T-1} (a \mu_1 + b \mu_2) \cdot \nstep{x_1}{i} \nstep{x_2}{i}  \right) - a \mathbb{E}(\nstep{x_1}{T})
	\\
	&\geq a \mathbb{E}(\nstep{x_1}{0}) + b \mathbb{E}(\nstep{x_2}{0}) + \mathbb{E} \left( \sum_{i = 0}^{T-1} (a \mu_1 + b \mu_2) \cdot \nstep{x_1}{i} \nstep{x_2}{i}  \right) - 10a
	\end{array}\]
	In particular, if $\mu_1 = \mu_2 = 0$, this simplifies to $M_k = a \nstep{x_1}{k} + b \nstep{x_2}{k}$ and thus, $a \mathbb{E}(\nstep{x_1}{T}) + b \mathbb{E}(\nstep{x_2}{T}) = a \mathbb{E}(\nstep{x_1}{0}) + b \mathbb{E}(\nstep{x_2}{0})$.
	In this case, $a = b$ or $a = -b$ holds.

	\subsection{A non-linear Markov system}
	
	\begin{minipage}{0.65\textwidth}
		Chakarov et al.~\cite{h} study a non-linear Markov system similar to the one on the right.
		This system has two states and in each state the two variables, $x_1$ and $x_2$, are updated according to the given statements.
		After every update of the variables this system stays in its current state with probability 0.5. With equal probability it transitions to the other state.
		
		This system can be modelled by a probabilistic transition system with one transition and no random variables. 
	\end{minipage}
	\hfill
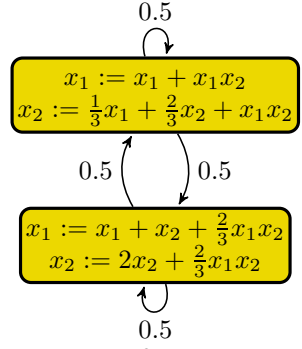
	\begin{minipage}{0.32\textwidth}
		\begin{tikzpicture}[->,>=stealth',shorten >=1pt,auto,node distance=2.0cm,
		semithick]
		
		\tikzset{every loop/.style={max distance=5mm}}
		
		\tikzset{
			state/.style={
				rectangle,
				rounded corners,
				fill={rgb:red,4;green,2;yellow,20},
				draw=black, very thick,
				minimum height=2em,
				minimum width=10em,
				inner sep=2pt,
				text centered,
			},
		}
		
		\node[state] 		 (A)	{$\begin{matrix}
			x_1 := x_1 + x_1 x_2\\
			x_2 := \frac{1}{3} x_1 + \frac{2}{3} x_2 + x_1 x_2
			\end{matrix}$};
		\node[state]         (B)     [below of = A] {$\begin{matrix}
			x_1 := x_1 + x_2 + \frac{2}{3} x_1 x_2\\
			x_2 := 2 x_2 + \frac{2}{3} x_1 x_2
			\end{matrix}$};
		
		\path 	(A) 	edge [loop above]	node 	{$0.5$} 	(A)
		(A)		edge [bend left]	node	{$0.5$}		(B)
		(B)		edge [loop below] 	node	{$0.5$}		(B)
		(B)		edge [bend left]	node	{$0.5$}		(A);	
		
		\end{tikzpicture}
	\end{minipage}
	Applying this method we find\footnote{Chakarov et al.~in \cite{h} found the same invariant.} that $V(x_1, x_2) = (x_1-x_2)^2$ satisfies the conditions $(a1)$ and $(a2')$ for $\alpha \leq 0.8$.
	Thus, the polynomial $P(x_1, x_2) = (x_1 - x_2)$ defines a probabilistic invariant via the martingale 
	\[
	M_k := 
	\left\{
	\begin{array}{ll} 
	\nstep{x_1}{0} - \nstep{x_2}{0} 
	&
	\hbox{if $k = 0$}\\
	\nstep{x_1}{0} - \nstep{x_2}{0} + \sum_{i = 1}^{k} \big(\nstep{x_1}{i} - \nstep{x_2}{i} - \mathrm{pre}\mathbb{E}(P)(\nstep{x_1}{i-1}, \nstep{x_2}{i-1})\big) 
	& \hbox{if $k > 0$} 
	\end{array} 
	\right.
	\]
	where $\nstep{x_1}{i}, \nstep{x_2}{i}$ are random variables representing the values of $x_1, x_2$ in the $i$-th step of the computation.
	
	A simple calculation shows
	$\mathrm{pre}\mathbb{E}(P)(x_1, x_2) = \frac{5}{6} \left(x_1 - x_2 \right).$ Hence
	\\
	\[
	M_k = 
	\left\{
	\begin{array}{ll} 
	\nstep{x_1}{0} - \nstep{x_2}{0} & \hbox{if $k = 0$}
	\\
	\nstep{x_1}{k} - \nstep{x_2}{k} + \frac{1}{6} \sum_{i = 0}^{k-1} (\nstep{x_1}{i} - \nstep{x_2}{i}) \quad & \hbox{if $k > 0$} 
	\end{array} 
	\right.
	\]
	Since $\{M_k\}_{k \in \mathbb{N}}$ is a martingale, for all $k \geq 0$,
	\[\begin{array}{rl}
	0 &= \mathbb{E}(M_{k+1})  - \mathbb{E}(M_k) =  \mathbb{E}(M_{k+1}  - M_k) 
	= \mathbb{E} (\nstep{x_1}{k+1} - \nstep{x_2}{k+1} - \frac{5}{6}(\nstep{x_1}{k} - \nstep{x_2}{k}))
	\\
	&= \mathbb{E} \big( (\nstep{x_1}{k+1} - \nstep{x_2}{k+1}) \big) - \mathbb{E} \big(\frac{5}{6}(\nstep{x_1}{k} - \nstep{x_2}{k}) \big).
	\end{array}\]
	By induction, we get the following invariant for this system:
	For all $k \in \mathbb{N}$,
	\(
	\mathbb{E}\big(\nstep{x_1}{k} - \nstep{x_2}{k}\big) = \big(\frac{5}{6}\big)^{k} \cdot \mathbb{E}\big(\nstep{x_1}{0} - \nstep{x_2}{0} \big) 
	\).
	In fact, we can go one step further.
	Straightforward computations show that 
	\[\textstyle
	Y_k := \big(\frac{6}{5}\big)^k \cdot (\nstep{x_1}{k} - \nstep{x_2}{k})
	\qquad \quad
	Z_k := \big(\frac{6}{5}\big)^k \cdot \big| \nstep{x_1}{k} - \nstep{x_2}{k} \big|
	\] 
	are martingales.
	These martingales satisfy the (IUD) precondition of the OST as 
	\(
	|Y_k| = |Z_k| \leq \sum_{i = 0}^{\infty} |Z_i|
	\)
	and
	\[\textstyle
	\mathbb{E}\big( \sum_{i = 0}^{\infty} |Z_i| \big) = \sum_{i = 0}^{\infty} \big(\frac{6}{5}\big)^i \cdot \mathbb{E}\big( |P(\nstep{x_1}{i}, \nstep{x_2}{i})| \big) 
	\leq \sum_{i = 0}^{\infty} \big(\frac{6}{5}\big)^i \cdot \alpha^i \cdot \mathbb{E} \big(V(\nstep{x_1}{0}, \nstep{x_2}{0}) \big)
	\]
	where the last inequality was shown in the proof of Lemma~\ref{lemma: |P| < infty}. 
	As $\alpha \leq 0.8$, i.e.~$\alpha \cdot \frac{6}{5} \leq 0.96$, the sum converges and we obtain indeed $\mathbb{E}\big( \sum_{i = 0}^{\infty} |Z_i| \big) < \infty$.
	\\
	Thus, for all stopping time $T: \Omega \rightarrow (\mathbb{N} \cup \{\infty\})$ the following equations hold.
	\begin{enumerate}[(i)]
		\item $ \mathbb{E} \big( \big( \frac{6}{5} \big)^T \cdot (\nstep{x_1}{T} - \nstep{x_2}{T}) \big) =  \mathbb{E}(Y_T) = \mathbb{E}(Y_0) = \mathbb{E}\big(\nstep{x_1}{0} - \nstep{x_2}{0} \big)$
		\item $ \mathbb{E} \big( \big( \frac{6}{5} \big)^T \cdot |\nstep{x_1}{T} - \nstep{x_2}{T}| \big) =  \mathbb{E}(Z_T) = \mathbb{E}(Z_0) = \mathbb{E} \big( |\nstep{x_1}{0} - \nstep{x_2}{0}| \big)$
	\end{enumerate}
	This is particularly interesting in the case of hitting times. For instance, we can consider the stopping time $T = \min \{ n \in \mathbb{N} \mid |\nstep{x_1}{n} - \nstep{x_2}{n}| = 0.5 \cdot \mathbb{E}(|\nstep{x_1}{0} - \nstep{x_2}{0}|)  \}$ for an initial distribution s.t.~$\mathbb{E}(|\nstep{x_1}{0} - \nstep{x_2}{0}|) > 0$.
	Then
	\[\textstyle
	\mathbb{E} \big( |\nstep{x_1}{0} - \nstep{x_2}{0}| \big)
	= 
	\mathbb{E} \big( \big( \frac{6}{5} \big)^T \cdot |\nstep{x_1}{T} - \nstep{x_2}{T}| \big)
	= 
	0.5 \cdot \mathbb{E}(|x_1^0 - \nstep{x_2}{0}|) \cdot \mathbb{E} \big( \big( \frac{6}{5} \big)^T \big)
	\]
	Hence 
	\(
	2 
	= 
	\mathbb{E} \big( \big( \frac{6}{5} \big)^T \big)
	= 
	\infty \cdot \mathbb{P} \big( T = \infty \big) + \sum_{n = 0}^{\infty} \big(\frac{6}{5}\big)^n \cdot \mathbb{P}\big( T = n \big)
	\).
	Firstly, this shows that $\mathbb{P}(T < \infty) = 1$. 
	Secondly, by the inequality
	$\big(\frac{6}{5}\big)^n \geq 1 + \frac{1}{5} n$,
	\[\begin{array}{ll}
	2 
	&= 
	\mathbb{E} \big( \big( \frac{6}{5} \big)^T \big)
	= 
	\sum_{n = 0}^{\infty} \big(\frac{6}{5}\big)^n \cdot \mathbb{P}\big( T = n \big) 
	\geq \sum_{n = 0}^{\infty} \big(1 + \frac{1}{5} \cdot n\big) \cdot \mathbb{P}( T = n)
	\\
	&= 
	1 + \frac{1}{5} \cdot \mathbb{E}(T).  
	\end{array}
	\]
	Therefore, $\mathbb{E}(T) \leq 5$ holds.
	By applying Jensen's Inequality \cite[Theorem 4.7.3]{borokov} we can further improve this bound to $\mathbb{E}(T) \leq 3.8 < 4$.
	%
	
	The above invariants hold for any stopping time $T$. 
	Hence it is possible to use this method to prove that programs are not a.s.~terminating.
	For instance consider the stopping time $T: = \min\left\lbrace n \in \mathbb{N} \mid \nstep{x_1}{n} = \nstep{x_2}{n} \right\rbrace$. 
	So $\nstep{x_1}{T} - \nstep{x_2}{T} = 0$.
	Nonetheless,
	$\mathbb{E} \big( \big( \frac{6}{5} \big)^T \cdot (\nstep{x_1}{T} - \nstep{x_2}{T}) \big) =  \mathbb{E}\big(\nstep{x_1}{0} - \nstep{x_2}{0} \big)$
	holds for $\nstep{x_1}{0} \neq \nstep{x_2}{0}$.
	This can only be, if $\mathbb{P}(T = \infty) > 0$.
	
	
	\section{Related work, conclusions and further directions} \label{sec: relatedwork}
	
	In recent years, there have been significant advances in two related problems: \emph{synthesis of probabilistic invariants}, and \emph{reasonsing about almost-sure termination}.
	(We present a survey of the myriads of methods and techniques for the two problems in Tables~\ref{table:invariants} and~\ref{figure: ast} respectively in \cite[Appendix]{SchreuderO19}.)
	
	That there is a deep link between the two problems is evident in the preceding, especially Sec.~\ref{sec: linear}.
	A telling perspective is provided by the OST (which underpins a recent approach to synthesizing invariants \cite{g}, 
	but) which typically requires the run-time to be bounded in some sense. 
	In fact one of the very few exceptions is our (IUD) precondition which seems new. 
	
	This link was recently investigated by {Hark et al.}~\cite{DBLP:journals/corr/abs-1904-01117}. 
	They show that in order to compute lower bounds on many properties of probabilistic programs, including run-time, it is necessary to satisfy the Optional Stopping Theorem. 
	In particular, any lower-bound approximation of the run-time is necessarily a probabilistic invariant in the sense of this paper. 
	We have constructed just such a lower bound in {Example}~\ref{eg: nested loop}. 
	
	In this paper, we have applied martingale theory to construct invariants.
	A conventional approach to reasoning about probabilistic programs uses weakest precondition, based on predicate transformer semantics \cite{e,wp1,DBLP:conf/cav/ChenHWZ15,DBLP:conf/atva/FengZJZX17,c,DBLP:conf/qest/GretzKM13,DBLP:conf/sas/KatoenMMM10,DBLP:conf/atva/FengZJZX17}. 
	As shown by Hark et al.~\cite{DBLP:journals/corr/abs-1904-01117}, these two methods are essentially equivalent.
	
	A related problem is the computation of moments of program properties.
	This is an important direction as moments capture such fundamental statistics as variance and skewness.
	Examples of program properties that have been so analysed include program variables \cite{forsyte} and run-time \cite{DBLP:conf/tacas/KuraUH19}.
	In the latter, Kura et al.~apply their algorithms to estimate the probability that a program will run exceptionally long. 
	Such quantities could be inferred using our geometric persistence method in Sec.~\ref{sec: geometric persistence}. 

	Barthe et al.~\cite{g} were the first to 
	apply the OST to automated program analysis.
	Given a ``seed'' polynomial in the program variables, Doob's Decomposition guarantees the existence of an associated martingale.
	They then use the (PDB) precondition of (Optional Stopping) Theorem~\ref{thm: OST} to ensure that the expected value of the polynomial upon termination equals that at the start of the computation.
	Our work in Sec.~\ref{sec: geometric persistence} follows in their footsteps, but we use the (IUD) precondition instead, which has the key advantage of not requiring (positive) almost-sure termination.
	Barthe et al.~consider a restricted class of probabilistic imperative programs (in which the loop body contains only assignments, and so, excluding nested loops). 
	Our setting of probabilistic transition systems is more general.
	
	Chakarov et al.~\cite{h} provided another inspiration for our work. 
	They have introduced 
	proof rules 
	for establishing almost-sure persistence 
	and recurrence 
	properties of 
	PTS.
	Their method can automatically infer supermartingales in the form of polynomial over the program variables, by applying SOS optimisation technique.
	They do not, however, consider probabilistic invariants \emph{qua} optional stopping martingales, which we do.

	
	
	
	\subsection*{Conclusion} \label{sec: conclusion}
	
	In this paper we have approached the problem of automatically synthesizing linear or polynomial invariants for probabilistic transition systems. In both cases we guarantee that the expected value upon termination is the same as the value at the start of the computation.
	
	For linear programs we have proven that this problem reduces to finding specific invariants of non-probabilistic programs.
	
	For general transition systems we have presented a methods to (automatically) synthesize polynomial invariants.
	These invariants are martingale expressions over the program variables.
	We have constructed a sum-of-squares constraint systems such that any feasible solutions leads to at least one probabilistic invariant.
	Concretely, this feasibility problem states sufficient conditions for the existence of geometric persistence properties. Based on the persistence properties we have applied the Optional Stopping Theorem \ref{thm: OST} in order to prove that the invariants are still valid after the program has terminated.
	
	We have implemented this method in MATLAB using the YALMIP \cite{yalmip2} and SeDuMi \cite{sedumi} optimisation solvers. Invariants could be found 
	within seconds on a laptop with an Intel(R) Core(TM) i7-6500 CPU @ 2.50 GHz and 8 GB RAM.
	
	\bibliographystyle{splncs04}
	\bibliography{references}
	
	
	%
	
	
	\clearpage
	
	
	\appendix
	
	
	\section{Tables of related work}
	
	\begin{table}[]
		\centering
		\begin{tabular}{|c|c|c|c|c|c|c|c|c|c|}
			\hline{}
			\multirow{2}{*}{} &
			\multirow{2}{*}{demonic?}
			& 
			{continuous} & 
			{type of} & nested & \multicolumn{3}{c|}{OST} & \multirow{2}{*}{tool} & \multirow{2}{*}{template}  
			\\ \cline{6-8}
			& 
			& distributions & invariant & loops & PDB & IUD & other\tablefootnote{Other preconditions of the OST} & &      
			\\ \hline
			Us
			& no & yes & linear & yes & yes & no & no & no & no\\ \hline
			Us
			& no & yes & polynomial & yes & no & yes & no & yes & degree
			\\ \hline
			\multirow{2}{*}{\cite{forsyte}}
			& \multirow{2}{*}{no} & \multirow{2}{*}{yes} & polynomial & \multirow{2}{*}{no} & \multirow{2}{*}{no} & \multirow{2}{*}{no} & \multirow{2}{*}{no} & \multirow{2}{*}{yes} & \multirow{2}{*}{no}
			\\
			& & & 
			and more & & & & & &  
			\\ \hline 
			\multirow{2}{*}{\cite{DBLP:journals/corr/abs-1902-04373}} & \multirow{2}{*}{yes} & No probabilistic & \multirow{2}{*}{polynomial} & \multirow{2}{*}{yes} & \multirow{2}{*}{no} & \multirow{2}{*}{no} & \multirow{2}{*}{no} & \multirow{2}{*}{yes} & \multirow{2}{*}{yes}
			\\ 
			& & nondeterminism & & & & & & &   
			\\ \hline
			\cite{DBLP:conf/atva/FengZJZX17}
			& no & yes & polynomial & yes & yes\tablefootnote{
				Although these papers do not express invariants as martingales nor make explicit use of the OST, they use preconditions that are essentially versions of the OST.
			}
			& no & yes\footnotemark[7] & yes & degree            
			\\ \hline
			\cite{DBLP:conf/cav/ChenHWZ15}
			& no & no & polynomial & no & yes\footnotemark[7] & no & yes\footnotemark[7] & yes & yes                   
			\\ \hline
			\cite{DBLP:conf/qest/GretzKM13}
			& yes & no & polynomial & no & yes\footnotemark[7] & no & yes\footnotemark[7] & yes & yes
			\\ \hline
			\cite{DBLP:conf/sas/KatoenMMM10}
			& yes & no & linear & yes & no & no & no & no & no
			\\ \hline
			\cite{g}
			& no & no & polynomial & no & yes & no & no & no & yes
			\\ \hline
			\multirow{2}{*}{\cite{DBLP:journals/corr/abs-1904-01117}}
			& \multirow{2}{*}{no} & \multirow{2}{*}{no} & positive & \multirow{2}{*}{yes} & \multirow{2}{*}{yes} & \multirow{2}{*}{no} & \multirow{2}{*}{yes} & \multirow{2}{*}{no} & \multirow{2}{*}{no}
			\\
			& & & functions & & & & & &
			\\ \hline
			\cite{b}
			& no & yes & linear & no & no & no & no & yes & yes
			\\ \hline
		\end{tabular}
		~\\[2mm]
		\caption{Methods for generating probabilistic invariants}\label{table:invariants}		
		%
		%
	\end{table}
	
	\begin{table}[]
		\centering
		\begin{tabular}{|c|c|c|c|c|c|c|c|c|c|}
			\hline
			\multirow{2}{*}{} &
			demonic
			& 
			{continuous} & 
			{type of} & nested & \multicolumn{3}{c|}{OST} & \multirow{2}{*}{tool} & \multirow{2}{*}{template}  
			\\ \cline{6-8}
			& nondeterminism & distributions & variant & loops & PDB & IUD & other & &      
			\\ \hline
			\multirow{3}{*}{\cite{e}}
			& \multirow{3}{*}{yes} & \multirow{3}{*}{yes} & positive & \multirow{3}{*}{yes} & \multirow{3}{*}{no} & \multirow{3}{*}{no} & \multirow{3}{*}{no} & \multirow{3}{*}{no} & \multirow{3}{*}{yes}
			\\
			& & & measurable & & & & & &
			\\
			& & & functions & & & & & &
			\\ \hline
			\cite{DBLP:journals/corr/abs-1901-06087}
			& yes & no & linear & yes & yes & no & no & yes & no
			\\ \hline
			\multirow{2}{*}{\cite{DBLP:conf/tacas/KuraUH19}}
			& \multirow{2}{*}{yes} & \multirow{2}{*}{yes} & linear and &\multirow{2}{*}{yes} & \multirow{2}{*}{no} & \multirow{2}{*}{no} & \multirow{2}{*}{no} & \multirow{2}{*}{yes} & \multirow{2}{*}{yes}
			\\			
			& & & polynomial & & & & & & 
			\\ \hline
			\cite{d}
			& yes & yes & linear & yes & no & no & no & yes & no
			\\ \hline
			\cite{a}
			& no & yes & linear & yes & no & no & no & yes & no
			\\ \hline
			\cite{c}
			& yes & yes  & linear & yes & no & no & yes  & no & no 
			\\ \hline
		\end{tabular}
		~\\[2mm]
		\caption{Methods for reasoning about positive a.s.~termination}\label{figure: ast}
	\end{table}

	\section{Supplementary materials for Section~\ref{sec: background}}
	\label{proof idf}
	
	
	Fix a probability space $(\Omega, \calF, \mathbb{P})$.
	Let $M = \{M_n\}_{n \in \mathbb{N}}$ be a sequence of random variables, 
	and $T: \Omega \rightarrow \mathbb{N} \cup \{\infty\}$ be a stopping time. 
	The \emph{process stopped at $T$}, $M^T = \{M_{n \wedge T}\}_{n \in \mathbb{N}}$, is defined by
	\[
	M_{n \wedge T}(\omega)
	:=
	\left\{
	\begin{array}{ll} 
	M_n(\omega) & \hbox{if $n \leq T(\omega)$}\\
	M_{T(\omega)}(\omega) & \hbox{otherwise}
	\end{array} \right.
	\]
	for all $n \in \mathbb{N}$ and $\omega \in \Omega$.		
	
	\begin{theorem}{\rm \cite[Theorem~p.~99]{Williams91}}
		\label{thm: williams}
		\begin{enumerate}[(i)]
			\item If $M = \{M_n\}_{n \in \mathbb{N}}$ is a supermartingale and $T$ is a stopping time, then the stopped process $M^T = \{M_{n \wedge T}\}_{n \in \mathbb{N}}$ is a supermartingale.
			In particular, for all $n \in \mathbb{N}$, $\mathbb{E}(M_{n \wedge T}) \leq \mathbb{E}(M_0)$.
			\item If $M = \{M_n\}_{n \in \mathbb{N}}$ is a martingale and $T$ is a stopping time, then the stopped process $M^T = \{M_{n \wedge T}\}_{n \in \mathbb{N}}$ is a martingale.
			In particular, for all $n \in \mathbb{N}$, $\mathbb{E}(M_{n \wedge T}) = \mathbb{E}(M_0)$.
		\end{enumerate}
	\end{theorem}	
	
	Let $X_1, X_2, \cdots$ and $X$ be random variables, and suppose they have finite $p$-th moments for some $p > 0$. 
		We say that $X_n$ \emph{converges to $X$ in $L^1$} if $\lim_{n \to \infty} \mathbb{E}(|X_n - X|^p) = 0$.
	Recall that a family $F \subseteq L^1$ is \emph{uniformly integrable} just if
	\begin{align*}
	\inf_{a \in [0, \infty)} \sup_{f \in F} \int_{\{f > a\}} |f| \ d\mathbb{P} = 0
	\end{align*}	
	
	\begin{theorem}{\rm \cite[Theorem 11.7]{klenke}}
		\label{thm: L1 convergence with uniform integrability}
		If $M = \{M_n\}_{n \in \mathbb{N}}$ be a uniformly integrable (sub/super)martingale adapted to a filtration $\calF = \{\calF_n\}_{n \in \mathbb{N}}$, 
		then there exists a $\calF_{\infty}$-measurable integrable random variable $X_{\infty}$ with $\lim_{n \rightarrow \infty} X_n = X_{\infty}$ a.s.~and in $L^1$.
		Furthermore:
		\begin{itemize}
			\item $X_n = \mathbb{E}\left[X_{\infty} \mid \calF_n \right]$ for all $n \in \mathbb{N}$, if $M$ is a martingale.	
			\item $X_n \leq \mathbb{E}\left[X_{\infty} \mid \calF_n \right]$ for all $n \in \mathbb{N}$, if $M$ is a submartingale.	
			\item $X_n \geq \mathbb{E}\left[X_{\infty} \mid \calF_n \right]$ for all $n \in \mathbb{N}$, if $M$ is a supermartingale.	
		\end{itemize}
	\end{theorem}
	
	\begin{lemma}{\rm \cite[Theorem 6.18(iii)]{klenke}}
		\label{lem: uniformly integrability inherits}
		Let $F, G \subseteq L^1$ be families of integrable random variables such that $G$ is uniformly integrable, and for all $f \in F$, there exists $g \in G$ with $|f| \leq |g|$.
		Then $F$ is also uniformly integrable.	
	\end{lemma}
	
	
	
	\subsection*{Proof of Theorem~\ref{thm: OST}: Precondition (IUD)}
	By Theorem \ref{thm: williams}, we have that the process stopped at $T$, $\{M_{n \wedge T}\}_{n \in \mathbb{N}}$, is a (super)martingale. 
	By assumption, $g$ is integrable; it follows that the family $\{g\} \subseteq L^1$ is uniformly integrable.
	Hence the assumptions of Lemma~\ref{lem: uniformly integrability inherits} are satisfied with $F = \{M_{n \wedge T} : n \in \mathbb{N} \}$ and $G = \{g\}$, and so, $\{M_{n \wedge T}\}_{n \in \mathbb{N}}$ is a uniformly integrable (super)martingale.
	By Theorem \ref{thm: L1 convergence with uniform integrability}, there is a random variable to which $M_{n \wedge T}$ converges almost surely and in $L^1$ as $n \rightarrow \infty$. 
	We already know that $M_{n \wedge T}$  converges point-wise to $M_{T}$.
	Hence by the additional $L^1$ convergence we can conclude
	\[
	\lim_{n \rightarrow \infty} \mathbb{E}(M_{n \wedge T}) = \mathbb{E}(M_{T}).
	\]		
	Furthermore, by Theorem~\ref{thm: williams}(ii), if $\{M_{n \wedge T}\}_{n \in \mathbb{N}}$ is a martingale, then for all $n \in \mathbb{N}$, 
	$\mathbb{E}(M_{0}) = \mathbb{E}(M_{n \wedge T})$,
	which gives us $\mathbb{E}(M_{0}) = \mathbb{E}(M_{T})$.
	Similarly, by Theorem~\ref{thm: williams}(i), if $\{M_{n \wedge T}\}_{n \in \mathbb{N}}$ is a supermartingale, then for all $n \in \mathbb{N}$, 
	$\mathbb{E}(M_{0}) \leq \mathbb{E}(M_{n \wedge T})$.
	As limits preserve inequalities, we can conclude
	\[
	\mathbb{E}(M_{0})
	= \lim_{n \rightarrow \infty} \mathbb{E}(M_{0})
	\leq
	\lim_{n \rightarrow \infty} \mathbb{E}(M_{n \wedge T}) = \mathbb{E}(M_{T}).
	\]	
	\qed
	
	\section{Supplementary materials for Section~\ref{sec: program}}
	
	\subsection{Additional example: Hare and Tortoise}
	
	\noindent\begin{minipage}{0.4\textwidth}
		\begin{algorithmic}
			\STATE {$x_1 := 0$} 
			\STATE {$x_2 := 30$} 
			\WHILE {$x_1 \leq x_2$}
			\STATE {$r_1 \sim \mathrm{uniform}(0, 1)$}
			\IF {$r_1 \geq 0.5$}
			\STATE {$r_2 \sim \mathrm{uniform}(0, 10)$}
			\STATE {$x_1 := x_1 + r_2$}
			\ENDIF
			\STATE {$x_2 := x_2 + 1$}
			\ENDWHILE
		\end{algorithmic}
	\end{minipage}%
	\begin{minipage}{0.6\textwidth}
		This example is from \cite{a}. It simulates the race between the hare ($x_1$) and the tortoise ($x_2$). The tortoise has a head start, but the hare can potentially run faster.
		
		Chakarov et al.~\cite{a} proved that this program terminates positively almost surely.
		Hence by Theorem~\ref{thm:correspondence}, if we find an invariant for the corresponding deterministic program, we have an invariant for the original one.
		
		\vfill
	\end{minipage}
	
	\medskip
	
	\noindent The deterministic equivalent reads:
	
	\begin{minipage}{0.45\textwidth}
		\begin{algorithmic}
			\STATE {$x_1 := 0$} 
			\STATE {$x_2 := 30$} 
			\WHILE {$x_1 \leq x_2$}
			\STATE {$r_2 := 5$}
			\STATE {$x_1 := 0.5 \cdot (x_1 + r_2) + 0.5 \cdot x_1$}
			\STATE {$x_2 := x_2 + 1$}
			\ENDWHILE	
		\end{algorithmic}
	\end{minipage}
	\hfill
	\begin{minipage}{0.45\textwidth}
		\begin{flushleft}
			Or~shorter,
		\end{flushleft}
		\begin{algorithmic}
			\STATE {$x_1 := 0$} 
			\STATE {$x_2 := 30$} 
			\WHILE {$x_1 \leq x_2$}
			\STATE {$x_1 := x_1 + 2.5$}
			\STATE {$x_2 := x_2 + 1$}
			\ENDWHILE	
		\end{algorithmic}
	\end{minipage}
	
	\medskip
	
	Hence (using Notation~\ref{notation}) $\{2\nstep{x_1}{n} - 5\nstep{x_2}{n}\}$, $\{\nstep{x_1}{n} - 2.5n\}$ and $\{\nstep{x_2}{n} + n\}$ are linear invariants for both programs.
	Let
	$T: \Omega \rightarrow \mathbb{N} \cup \{\infty\}$ be the exit time of the loop, i.e.~$\mathbb{E}\big(\nstep{x_1}{T}\big) \geq \mathbb{E}\big(\nstep{x_2}{T}\big)$.
	Since
	\(
	-150 = \mathbb{E}(2\nstep{x_1}{0} - 5 \nstep{x_2}{0}) = \mathbb{E}\big(2\nstep{x_1}{T} - 5 \nstep{x_2}{T} \big) \geq \mathbb{E}\big(2\nstep{x_2}{T} - 5 \nstep{x_2}{T} \big) = -3 \mathbb{E} \big(\nstep{x_2}{T}\big),
	\)
	we have $\mathbb{E} \big(\nstep{x_1}{T}\big) \geq \mathbb{E} \big(\nstep{x_2}{T}\big) \geq 50$.
	Now, $0 = \mathbb{E} \big(\nstep{x_1}{0} - 2.5 \cdot 0 \big) = \mathbb{E} \big(\nstep{x_1}{T} - 2.5 \cdot T \big)$; it follows that $\mathbb{E}\big(\nstep{x_1}{T}\big) = 2.5 \cdot \mathbb{E} \big( T \big)$.
	Similarly, because $30 = \mathbb{E} \big(\nstep{x_2}{0} -  0 \big) = \mathbb{E} \big(\nstep{x_2}{T} - T \big)$, we have $\mathbb{E}\big(\nstep{x_2}{T}\big) = 30 + \mathbb{E} \big( T \big)$. 
	Thus, $\mathbb{E} \big(T \big) \geq 20$.		

	\subsection{Additional proofs}
	
	\begin{lemma}
		\label{lemma: pts,measurable}
		Let $\Pi$ be a PTS and $\{X^k, L^k\}$ be the operational semantics of $\Pi$. 
		Let $F_k = \sigma(X_0, \widetilde{R^1}, \widetilde{R^2}, \cdots, \widetilde{R^k})$ denote our standard filtration.
		Then $(X_k,L_k)$ is $\mathbb{B}(\mathbb{R}^n) \otimes \mathcal{P}(L)$\footnote{
			$\mathcal{P}$ denotes the powerset.
		}-$\calF_k$-measurable.
	\end{lemma}
	
	\begin{proof} 
		We prove the claim by induction.
		
		For the base case: $k$ = 0:
		As $\calF_0 = \sigma(X_0)$,  $X_0$ is $\mathbb{B}(\mathbb{R}^n)$-$F_0$-measurable.
		As a constant function, $L_0$ is trivially $\mathcal{P}(L)$-$\calF_0$-measurable.
		
		For the inductive case, we have, by the induction hypothesis:
		\begin{center}
			\begin{tabular}{lll}
				$g_k$: &$\Omega$ &$\rightarrow \mathbb{R}^n \times L$\\
				&$\omega$ &$\mapsto (X_k(\omega),L_k(\omega))$
			\end{tabular}
		\end{center}
		is $\mathbb{B}(\mathbb{R}^n) \otimes \mathcal{P}(L)$-$\calF_k$-measurable.\\
		$\Pi$ is non-demonic. Thus $\tau_{k+1}$ is uniquely determined by $X^k, L^k$ based on polynomial guards on $\mathbb{R}^n$.
		This can be represented by a $\mathcal{P}(T)$-$\mathbb{B}(\mathbb{R}^n) \otimes \mathcal{P}(L)$-mea\-sur\-able selection function $\psi$.
		Therefore, $\tau_{k+1} = \psi \circ g_k$ is 
		$\mathcal{P}(T)$-$\calF_k$-measurable.
		$L^{k+1}$ is entirely defined by $\tau_{k+1}$.
		Hence $L^{k+1}$ is $\mathcal{P}(L)$-$\calF_k$-measurable.
		
		By Definition \ref{def: pts},
		$X^{k+1} = f_{\tau_{k+1}}(X^k, \widetilde{R}^{k+1}) = F_{\tau_{k+1},i}(X^k, \widetilde{R}^{k+1})$
		for a randomly chosen $i \in \{1,2,\cdots,j_{\tau_{k+1}}\}$ with probability $p_{\tau_{k+1},i} > 0$.
		The choice of $i$ 
		is contained in  $\widetilde{R}^{k+1}$ and hence $\calF_{k}$-measurable.
		$F_{\tau_{k+1},i}$ is continuous.
		As a consequence, 
		$X^{k+1} = F_{\tau_{k+1},i}(X^k, \widetilde{R}^{k+1})$ is
		$\mathbb{B}(\mathbb{R}^{n})$-$\calF_{k+1}$-measurable.
		$L^{\psi}_{k+1} = l_{\tau_{k+1},i}$ does only on $\tau_{k+1}$ and the choice of $i$.
		So $L_{k+1}$ is $\mathcal{P}(L)$-$\calF_{k+1}$-measurable and is therefore $\mathcal{P}(L)$-$\calF_{k+1}$-measurable. 
		Hence $(X^{k+1}, L^{k+1})$ is $\mathbb{B}(\mathbb{R}^{n}) \otimes \mathcal{P}(L)$-$\calF_k$-measurable.	
		\qed		
	\end{proof}
	
	%
	
	
	\section{Supplementary materials for Section~\ref{sec:invariant}}	
	
	\lemmatilde
	
	\begin{proof} 
		Let $k \in \mathbb{N}$.
		\[\begin{array}{cl}
		&\mathbb{E}[h(X^{k+1}, L^{k+1}, k+1) \mid \calF_k]
		\\
		= & 
		\mathbb{E}\Big[\Big( \sum_{i = 1}^p [(X^k \models \phi_i) \wedge (L^k = \sloc{\tau_i})] \Big) \cdot h(X^{k+1}, L^{k+1}, k+1) \mid \calF_k \Big]
		\\
		= &
		\mathbb{E}\Big[\Big( \sum_{i = 1}^p [(X^k \models \phi_i) \wedge (L^k = \sloc{\tau_i})] \Big) \cdot h(f_{\tau_i}(X^{k}, R^{k+1}), \dloc{\tau_{i}}, k+1) \mid \calF_k \Big] 
		\\
		= &
		\sum_{i = 1}^p [(X^k \models \phi_i) \wedge (L^k = \sloc{\tau_i})] \Big) \cdot
		\mathbb{E}\Big[  h(f_{\tau_i}(X^{k}, R^{k+1}), \dloc{\tau_{i}}, k+1) \mid \calF_k \Big] 
		\\
		&\hspace*{12pt}
		X^k, L^k \mathrm{\ are \ } \calF_{k}\mathrm{-measurable}
		\end{array}\]
		%
		
		Moreover, as $h(-, \dloc{\tau_i}, -) = P_{i}$ is a polynomial, we have
		\[\begin{array}{ll}
		& \mathbb{E}\Big[  h(f_{\tau_i}(X^{k}, R^{k+1}), \dloc{\tau_{i}}, k+1) \mid \calF_k \Big]
		= 
		\mathbb{E}\Big[  P_i(f_{\tau_i}(X^{k}, R^{k+1}),  k+1) \mid \calF_k \Big]
		\\
		= &
		\mathbb{E}\Big[ \Big(
		\sum_{q = 1}^{j_{\tau_i}} [f_{\tau_i} = F_{\tau_i, q}] \Big) \cdot P_i(f_{\tau_i}(X^k, R^{k+1}), k+1) \mid \calF_k \Big].
		\end{array}
		\]
		The functions
		$P_i$ and $F_{\tau_i, q}$ are polynomials in $X^k$, $R^{k+1}$ and $k$.
		$X^k$ are $\calF_k$-measurable. The random choice $[f_{\tau_i} = F_{\tau_i, q}]$ and $R^{k+1}$ are $\calF_k$-independent. 
		The conditional expectation is linear. Hence we can split the conditional expectation above in a linear combination of conditional expectations of $[f_{\tau_i} = F_{\tau_i, q}]$ multiplied with monomials over $X^k$, $R^{k+1}$ and $k$. Monomials over $X^k$ are $\calF_k$-measurable and $k$ is a constant, so we can pull them out of the conditional expectations. 
		The rest is in\-de\-pen\-dent of $\calF_k$. Hence the remaining conditional expectation is an ordinary expectation. By mutual independence we can split the expectations over $R^{k+1}$ and $\mathbb{E} [f_{\tau_i} = F_{\tau_i, q}]$.
		Note that $\mathbb{E} [f_{\tau_i} = F_{\tau_i, q}] = p_{\tau_i, q}$.
		Hence
		\[\begin{array}{ll}
		&\mathbb{E}\Big[  h(f_{\tau_i}(X^{k}, R^{k+1}), \dloc{\tau_{i}}, k+1) \mid \calF_k \Big]
		\\
		= &
		\sum_{q = 1}^{j_{\tau_i}}
		p_{\tau_i, q} \cdot
		\mathbb{E}_{R^{k+1}} \left( P_i(F_{\tau_i, q}(X^k, R^{k+1}),k+1) \right)
		=
		\mathrm{pre}\mathbb{E}(h, \tau_i)(X^k, k+1)
		\end{array}\]
		which proves the claim.
		
		As for part (ii), $\mathbb{E}[h(X^{k+1}, L^{k+1}, k+1) \mid \calF_k] = \mathrm{pre}\mathbb{E}(P)(X^k, L^k, k+1) \leq h(X^k, L^k, k)$ by part 1.
		\qed
	\end{proof}

	\section{Supplementary materials for Section~\ref{sec: linear}}
	
	\linearOST
	
	
	\begin{proof}
		(i): Let $x \in \mathbb{R}^n$ and $k \in \mathbb{N}$. We have:
		\[		\begin{array}{lll}
		\mathrm{pre}\mathbb{E}(h, \tau)(x, k)
		&= 
		\sum_{i = 1}^{j_{\tau}} p_{\tau, i} \cdot 
		\mathbb{E}_R(h(F_{\tau, i}(x, R), \dloc{\tau}, k))
		\\
		&= 
		\sum_{i = 1}^{j_{\tau}} p_{\tau, i} \cdot 
		\mathbb{E}_R(
		P_{\dloc{\tau}} \left(
		F_{\tau, i}(x, R), k \right))
		\\
		&=  
		P_{\dloc{\tau}} \left(
		\sum_{i = 1}^{j_{\tau}} p_{\tau, i} \cdot
		\mathbb{E}_R( F_{\tau, i}(x, R)),
		\sum_{i = 1}^{j_{\tau}} p_{\tau, i} \cdot k \right)
		\quad 
		P_{\dloc{\tau}} \mathrm{\ linear}	
		\\ 
		&=  
		P_{\dloc{\tau}} \left(
		\sum_{i = 1}^{j_{\tau}} p_{\tau, i} \cdot
		F_{\tau, i}(x, \mu), k \right)
		\hspace*{80pt} F_{\tau, i} \mathrm{\ linear}	
		\\
		&=  
		P_{\dloc{\tau}} \left(
		f_{\tau}^{\mathsf{det}}(x), k \right)
		= h \left(
		f_{\tau}^{\mathsf{det}}(x), \dloc{\tau}, k \right)
		= \mathrm{pre}\mathbb{E}^{\mathsf{det}}(h, \tau)(x, k)
		\end{array}\]
		
		(ii):
		\begin{align*}
		\begin{split}
		&\mathbb{E}[| h(X^{k+1}, L^{k+1}, k+1) - h(X^{k}, L^k, k) | \mid \calF_{k}]
		\\
		= \ &
		\mathbb{E} 
		\left[
		\left| 
		\sum_{i = 1}^p 
		[(X^k \models \phi_i) \wedge (L^k = \sloc{ \tau_i})]
		\left(
		P_{\dloc{\tau_{i}}}(f_{\tau_i}(X^k, R^{k+1}), k+1) 
		\right. \right. \right.
		\\
		&\left. \left. \left.
		- \ h(X^{k}, L^k, k) \right) \right|
		\mid \calF_{k}\right]	
		\\
		=&
		\sum_{i = 1}^p 
		[(X^k \models \phi_i) \wedge (L^k = \sloc{ \tau_i})]
		\cdot
		\mathbb{E}[
		|   P_{\dloc{\tau_{i}}} (f_{\tau_i}(X^k, R^{k+1}), k+1) 
		\\
		&- h(X^{k}, L^k, k)  |
		\mid \calF_{k}] 	
		\\
		\leq
		&\sum_{i = 1}^p 
		[(X^k \models \phi_i) \wedge (L^k = \sloc{\tau_i})]
		\mathbb{E}[
		|   P_{\dloc{\tau_{i}}} (f_{\tau_i}(X^k, \mu), k+1)
		- h(X^{k}, L^k, k)  |
		\mid \calF_{k}] 
		\\
		&+
		\sum_{i = 1}^p 
		[(X^k \models \phi_i) \wedge (L^k = \sloc{ \tau_i})]
		\cdot
		\mathbb{E}[
		|   P_{\dloc{\tau_{i}}}(f_{\tau_i}(X^k, R^{k+1}), k+1) 
		\\
		&- P_{\dloc{\tau_{i}}, k} (f_{\tau_i}(X^k, \mu))  |
		\mid \calF_{k}]
		\end{split}
		\end{align*}	
		Considering the first summand we obtain for $\tau_i \in T$ fixed,
		\begin{align*}
		\begin{split}
		&\mathbb{E}[
		|   P_{\dloc{\tau_{i}}} (f_{\tau_i}(X^k, \mu), k+1) - h(X^{k}, L^k, k)  |
		\mid \calF_{k}]
		\\
		=& 
		\mathbb{E} \left[
		\left| \sum_{q = 1}^{j_{\tau_i}} [f_{\tau_i} = F_{\tau_i, q}] \left( P_{\dloc{\tau_{i}}} (F_{\tau_i, q}(X^k, \mu), k+1) - h(X^{k}, L^k, k) \right)  \right|
		\mid \calF_{k} \right]
		\\
		=&
		\sum_{q = 1}^{j_{\tau_i}}
		\mathbb{E}[
		|   [f_{\tau_i} = F_{\tau_i, q}]   |] \cdot
		\left| P_{\dloc{\tau_{i}}} (F_{\tau_i, q}(X^k, \mu), k+1) - h(X^{k}, L^k, k) \right|
		\\
		=&
		\sum_{q = 1}^{j_{\tau_i}}
		p_{\tau_i, q} \cdot
		\left| P_{\dloc{\tau_{i}}} (F_{\tau_i, q}(X^k, \mu), k+1) - h(X^{k}, L^k, k) \right|
		\\
		=&
		\left| \sum_{q = 1}^{j_{\tau_i}}
		p_{\tau_i, q} \cdot P_{\dloc{\tau_{i}}} (F_{\tau_i, q}(X^k, \mu), k+1) - h(X^{k}, L^k, k) \right|
		\\
		=&
		\left| 
		\mathrm{pre}\mathbb{E}^{\mathsf{det}}(h, \tau_{i}) (X^k, k+1) - h(X^{k}, L^k, k) \right|.
		\end{split}
		\end{align*}
		Therefore, we can conclude
		\begin{align*}
		\begin{split}
		&\sum_{i = 1}^p 
		[(X^k \models \phi_i) \wedge (L^k = \sloc{ \tau_i})]
		\ 
		\mathbb{E}[
		|   P_{\dloc{\tau_{i}}} (f_{\tau_i}(X^k, \mu), k+1) - h(X^{k}, L^k, k)  |
		\mid \calF_{k}]
		\\
		=&
		\sum_{i = 1}^p 
		[(X^k \models \phi_i) \wedge (L^k = \sloc{ \tau_i})]
		\cdot
		\left| 
		\mathrm{pre}\mathbb{E}^{\mathsf{det}}(h, \tau_{i}) (X^k, k+1) - h(X^{k}, L^k, k) \right|
		\\
		=&
		\left| 
		\sum_{i = 1}^p 
		[(X^k \models \phi_i) \wedge (L^k = \sloc{ \tau_i})]
		\cdot
		\mathrm{pre}\mathbb{E}^{\mathsf{det}}(h, \tau_{i}) (X^k, k+1) - h(X^{k}, L^k, k) \right|
		\\
		=&
		\left|
		\mathrm{pre}\mathbb{E}^{\mathsf{det}}(h) (X^k, L^k, k+1) - h(X^{k}, L^k, k) \right|.
		\end{split}
		\end{align*}
		The second term is in fact bounded by a constant.
		For $\tau_i \in T$ fixed, we have
		\begin{align*}
		\begin{split}
		&\mathbb{E}[
		|   P_{\dloc{\tau_{i}}}(f_{\tau_i}(X^k, R^{k+1}), k+1) - P_{\dloc{\tau_{i}}} (f_{\tau_i}(X^k, \mu), k+1)  |
		\mid \calF_{k}]
		\\
		=&
		\mathbb{E} \left[
		\left| 
		\sum_{q = 1}^{j_{\tau_i}} [f_{\tau_i} = F_{\tau_i, q}] \cdot
		\left(  P_{\dloc{\tau_{i}}}(F_{\tau_i, q}(X^k, R^{k+1}), k+1) 
		\right. \right. \right.
		\\
		%
		&- P_{\dloc{\tau_{i}}, k} (F_{\tau_i, q}(X^k, \mu), k+1) ) | 
		\mid \calF_{k} ] 
		\\
		& \mathrm{The \ functions \ } P_{\dloc{ \tau_{i}}}, F_{\tau_i, q} \mathrm{ \ are \ linear. \ Hence \ } P_{\dloc{\tau_{i}}} \circ F_{\tau_i, q} 
		\mathrm{ \ is \ linear \ as \ well.}  
		\\
		&\mathrm{ Thus, \ the \ term \ inside \ the \ conditional \ expectation }
		\mathrm{ \ does \ not \ depend \ on \  X^k. }
		\\
		=&
		\mathbb{E} \left[
		\left| 
		\sum_{q = 1}^{j_{\tau_i}} [f_{\tau_i} = F_{\tau_i, q}] \cdot
		\left(  P_{\dloc{ \tau_{i}}}(F_{\tau_i, q}(X^k, R^{k+1} - \mu), k+1) 
		\right) \right| 		\right] \in \mathbb{R}.
		\end{split}
		\end{align*}
		As we only have finitely many transitions, this concludes the proof.
		
		\qed
	\end{proof}
	
	\section{Supplementary materials for Section~\ref{sec: geometric persistence} }
	
	In this Section, we fix a PTS $\Pi = \langle W, X, R, X_0, L, l_0, l_F, T \rangle$,
	and assume functions $P,  V: \mathbb{R}^n \times L \rightarrow \mathbb{R}$ such that for all $l \in L$, $V(-, l)$ and $P(-, l)$ are polynomials satisfying $V \geq P^2$.
	
	\begin{lemma}
		\label{lemma: diagonal terms} 
		Let $\tau \in T$ be a transition.
		Then 
		$\mathrm{pre}\mathbb{E}(P,\tau)(z)^2 \leq j_{\tau} \cdot \mathrm{pre}\mathbb{E}(V,\tau)(z)$ for all $z \in \mathbb{R}^n$,
		where $j_{\tau} \in \mathbb{N}$ depends only on the transition $\tau$.
	\end{lemma}
	
	\begin{proof} 
		Let $z \in \mathbb{R}^n$.
		\[\begin{array}{ll}
		&\mathrm{pre}\mathbb{E}(P,\tau)(z)^2\\
		= & 
		\left( \sum_{i = 1}^{j_{\tau}} p_{\tau,i} \cdot \mathbb{E}_R(P(F_{\tau,i}(z,R), \dloc{\tau})) \right)^2
		\\
		= & 
		\sum_{i,k = 1}^{j_{\tau}} p_{\tau,i} \cdot p_{\tau,k} \cdot \mathbb{E}_R(P(F_{\tau,i}(z,R), \dloc{\tau}) \cdot \mathbb{E}_R(P(F_{\tau,k}(z,R), \dloc{\tau}) 
		\\ 
		\leq & 
		\sum_{i,k = 1}^{j_{\tau}} 
		\frac{1}{2} \cdot 
		\left( p_{\tau,i}^2 \cdot \mathbb{E}_R(P(F_{\tau,i}(z,R), \dloc{\tau})^2 +
		p_{\tau,k}^2 \cdot \mathbb{E}_R(P(F_{\tau,k}(z,R), \dloc{\tau})^2 \right)
		\\ 
		= &
		\sum_{i, k = 1}^{j_{\tau}} 
		\frac{1}{2} \cdot 
		p_{\tau,i}^2 \cdot \mathbb{E}_R(P(F_{\tau,i}(z,R), \dloc{\tau}))^2 
		+
		\sum_{i, k = 1}^{j_{\tau}}  
		\frac{1}{2} \cdot
		p_{\tau,k}^2 \cdot \mathbb{E}_R(P(F_{\tau,k}(z,R), \dloc{\tau}))^2 
		\\ 
		= & 
		j_{\tau} \cdot  \sum_{i = 1}^{j_{\tau}} p_{\tau,i}^2 \cdot \mathbb{E}_R(P(F_{\tau,i}(z,R), \dloc{\tau}))^2
		\leq 
		j_{\tau} \cdot  \sum_{i = 1}^{j_{\tau}} p_{\tau,i} \cdot \mathbb{E}_R(P(F_{\tau,i}(z,R), \dloc{\tau}))^2 
		\\ 
		\leq & 
		j_{\tau} \cdot  \sum_{i = 1}^{j_{\tau}} p_{\tau,i} \cdot \mathbb{E}_R(P(F_{\tau,i}(z,R), \dloc{\tau})^2)
		\leq 
		j_{\tau} \cdot \mathrm{pre}\mathbb{E}(V, \tau)(z)
		\end{array}
		\]
		As all moments of $R$ exists, $\mathbb{E}_{R}(P(F_{\tau,i}(z,R), \dloc{\tau})^2) < \infty$.
		Therefore, the second-to-last inequality is  Jensen's Inequality (\cite[Theorem 4.7.3]{borokov}).
		\qed
	\end{proof}

	\begin{lemma}
		\label{lemma: pre P2 leq pre V} 
		For all $z \in \mathbb{R}^n$ and $l \in L$,
		$\mathrm{pre}\mathbb{E}(P)(z, l)^2 \leq C_{\Pi} \cdot \mathrm{pre}\mathbb{E}(V)(z, l) $
		where $C_{\Pi} \in \mathbb{R}$ depends only on the PTS $\Pi$.
	\end{lemma} 
	
	\begin{proof}
		Let $z \in \mathbb{R}^n$ and $l \in L$.
		\[\begin{array}{ll}
		&\mathrm{pre}\mathbb{E}(P)(z, l)^2 \\
		= & 
		\left( \sum_{i = 1}^p [(z \models \phi_i) \wedge (l = \sloc{\tau})] \cdot \mathrm{pre}\mathbb{E}(P,\tau_i)(z) \right)^2
		\\ 
		= &
		\sum_{i, j = 1}^p [(z \models \phi_i) \wedge (l = \sloc{\tau})] [(z \models \phi_j) \wedge (l = \sloc{\tau})]  \mathrm{pre}\mathbb{E}(P,\tau_i)(z)  \mathrm{pre}\mathbb{E}(P,\tau_j)(z) 
		\\
		= &
		\sum_{i = 1}^p [(z \models \phi_i) \wedge (l = \sloc{\tau})] \cdot \mathrm{pre}\mathbb{E}(P,\tau_i)(z)^2
		\hspace*{20pt}
		\hbox{($\Pi$ is non-demonic)}
		\\ 
		\leq &
		\sum_{i = 1}^p [(z \models \phi_i) \wedge (l = \sloc{\tau})] \cdot j_{\tau_i} \cdot \mathrm{pre}\mathbb{E}(V,\tau_i)(z)
		\hspace*{20pt}
		\hbox{(Lemma \ref{lemma: diagonal terms})}
		\\ 
		\leq & 
		\max_{\tau \in T}\{j_{\tau}\} \cdot  \sum_{i = 1}^p [(z \models \phi_i) \wedge (l = \sloc{\tau})] \cdot \mathrm{pre}\mathbb{E}(V,\tau_i)(z)
		\\
		= & 
		\max_{\tau \in T}\{j_{\tau}\} \cdot \mathrm{pre}\mathbb{E}(V)(z) 
		\end{array}\]
		\qed
	\end{proof}

	\begin{lemma} \label{lemma: (P - pre P)^2 leq V} 
		For all $z \in \mathbb{R}^n$, and for all
		$l \in L$
		\[
		\big(P(z, l) - \mathrm{pre}\mathbb{E}(P)(z, l)\big)^2 \leq (2 + 2 \cdot C_{\Pi}) \cdot V(z, l)
		\]
		where $C_{\Pi} \in \mathbb{R}$ depends only on $\Pi$.
	\end{lemma}
	\begin{proof} 
		Take $z \in \mathbb{R}^n$ and $l \in L$. Then
		\[\begin{array}{lll}
		& (P(z, l) - \mathrm{pre}\mathbb{E}(P)(z, l))^2\\
		= & 
		P(z, l)^2 - 2 \cdot P(z, l) \cdot \mathrm{pre}\mathbb{E}(P)(z, l) + \mathrm{pre}\mathbb{E}(P)(z, l)^2 
		\\
		\leq & 
		P(z, l)^2 + 2  |P(z, l) \cdot \mathrm{pre}\mathbb{E}(P)(z, l)| + \mathrm{pre}\mathbb{E}(P)(z, l)^2
		\\
		\leq & 2 \cdot P(z, l)^2 +  2 \cdot \mathrm{pre}\mathbb{E}(P)(z, l)^2 
		\\ 
		\leq &
		2 \cdot V(z, l) + 2 \cdot C_{\Pi} \cdot \mathrm{pre}\mathbb{E}(V)(z, l)
		& \textrm{(Lemma \ref{lemma: pre P2 leq pre V})}
		\\ 
		\leq &
		2 \cdot V(z, l) + 2 \cdot C_{\Pi} \cdot \alpha \cdot V(z, l)
		& \textrm{(a2)}
		\\
		\leq &
		(2 + 2 \cdot C_{\Pi}) \cdot V(z, l)
		\end{array}
		\] 
		\qed	
	\end{proof}
	Now we show that $\sum_{i \in \mathbb{N}} V(X^i, L^i)$ converges.
	
	\begin{lemma} \label{lemma: geometric series}
		$\mathbb{E}(V(X^k, L^k)) \leq \alpha^k \cdot \mathbb{E}(V(X^0, L^0))$ for all $k \in \mathbb{N}$.
	\end{lemma}
	\begin{proof} We prove this by induction on $k$.
		So
		$\mathbb{E}(V(X^0, L^0)) = \alpha^0 \cdot \mathbb{E}(V(X^0, L^0))$.
		By (a2) there is $\alpha \in [0,1)$ such that $\mathrm{pre}\mathbb{E}(V)(x, l) \leq \alpha \cdot V(x, l)$ for all $x \in \mathbb{R}^n$ and $l \in L$.
		By Lemma \ref{lemma tilde}, $\mathbb{E}[V(X^{k+1}, L^{k+1}) \mid \calF_k] = \mathrm{pre}\mathbb{E}(V)(X^k, L^k) \leq \alpha \cdot V(X^k, L^k)$, i.e.~$(V(X^k, L^k))_{k}$ is an $\alpha$-multiplicative supermartingale w.r.t.~${\calF_k}$.
		Using the Tower Property of conditional expectation, and the induction hypothesis, we have:
		\begin{align*}
		\mathbb{E}(V(X^{k+1}, L^{k+1})) 
		&= \mathbb{E}(\mathbb{E}[V(X^{k+1}, L^{k+1}) \mid \calF_k]) 
		\\
		&\leq \mathbb{E}(\alpha \cdot V(X^k, L^k)) 
		\leq \alpha^{k+1} \cdot \mathbb{E}(V(X^0, L^0)).
		\end{align*}
		\qed
	\end{proof}
	
	%
	
	\begin{lemma}\label{lem:exp-inf-sum-p-preEP}
		$\mathbb{E} \left( \sum_{i = 0}^{\infty} |P(X^i, L^i) - \mathrm{pre}\mathbb{E}(P)(X^i, L^i)| \right) < \infty$.
	\end{lemma}
	
	\begin{proof} 
		Let $i \in \mathbb{N}$.
		\begin{align*}
		&\mathbb{E}(|P(X^i, L^i) - \mathrm{pre}\mathbb{E}(P)(X^i, L^i)|^2) 
		\leq 
		\mathbb{E}((2 + 2  C_{\Pi}) \cdot V(X^i, L^i)) 
		\hspace*{18pt}
		\textrm{(Lem. \ref{lemma: (P - pre P)^2 leq V})}
		\\
		&= 
		(2 + 2  C_{\Pi}) \cdot \mathbb{E}(V(X^i, L^i))
		\leq 
		(2 + 2  C_{\Pi}) \cdot \alpha^i \cdot \mathbb{E}(V(X^0, L^0)) 
		< \infty
		\ \textrm{(Lem. \ref{lemma: geometric series})} 
		\end{align*}	
		By the Hölder inequality (\cite[Theorem 4.7.3]{borokov}), we can conclude 
		\[
		\mathbb{E}(|P(X^i, L^i) - \mathrm{pre}\mathbb{E}(P)(X^i, L^i)|) \leq (2 + 2 \cdot C_{\Pi}) \cdot \alpha^i \cdot \mathbb{E}(V(X^0, L^0)).
		\]	
		Thus
		\[
		\begin{array}{ll}
		&\mathbb{E} \left( \sum_{i = 0}^{\infty} |P(X^i, L^i) - \mathrm{pre}\mathbb{E}(P)(X^i, L^i)| \right)
		\\
		=&
		\sum_{i = 0}^{\infty}  \mathbb{E} \left(|P(X^i, L^i) - \mathrm{pre}\mathbb{E}(P)(X^i, L^i)| \right)
		\\ 
		\leq &
		\sum_{i = 0}^{\infty} (2 + 2 C_{\Pi})  \alpha^i  \mathbb{E}(V(X^0, L^0))
		\\
		= & 
		(2 + 2 C_{\Pi}) \cdot  \mathbb{E}(V(X^0, L^0)) \cdot \sum_{i = 0}^{\infty} \alpha^i
		\\
		= & 
		(2 + 2 C_{\Pi}) \cdot  \mathbb{E}(V(X^0, L^0)) \cdot \frac{1}{1-\alpha} 
		< \infty.
		\end{array}
		\]
		\qed
	\end{proof}
	
	\begin{lemma}
		\label{lemma: |P| < infty}
		$\mathbb{E} \left( \sum_{i = 0}^{\infty} |P(X^i, L^i)| \right) < \infty$.
	\end{lemma}
	
	\begin{proof} 
		By Lemma \ref{lemma: geometric series}, for all $i \in \mathbb{N}$:
		\[
		\mathbb{E}(|P(X^i, L^i)|^2) = \mathbb{E}(P(X^i, L^i)^2) \leq \mathbb{E}(V(X^i, L^i)) \leq \alpha^i \cdot \mathbb{E}(V(X^0, L^0)).
		\] 
		By Hölder's Inequality (\cite[Theorem 4.7.3]{borokov}), $\mathbb{E}(|P(X^i, L^i)|) \leq \alpha^i \cdot \mathbb{E}(V(X^i, L^i))$ for all $i \in \mathbb{N}$.
		It follows from that
		\[
		\begin{array}{rll}
		\mathbb{E} \left( \sum_{i = 0}^{\infty} |P(X^i, L^i)| \right)
		&=& 
		\sum_{i = 0}^{\infty} \mathbb{E} \left( |P(X^i, L^i)| \right)
		\\
		&\leq&  
		\sum_{i = 0}^{\infty} \alpha^i \cdot \mathbb{E}(V(X^0, L^0))
		= 
		\dfrac{\mathbb{E}(V(X^0))}{1-\alpha} < \infty.
		\end{array}
		\]
		\qed
	\end{proof}


	\begin{proof}
		
		As $f$ is convex at $\mathbb{E}(X) \in \mathbb{R}$, there is a linear function $g: \mathbb{R} \rightarrow \mathbb{R}$ such that $f(\mathbb{E}(X)) = g(\mathbb{E}(X))$ and $f(y) \geq g(y)$ for all $y \in \mathbb{R}$.
		Thus by the linearity and the monotonicity of the expected value:
		\\
		$\Rightarrow f(\mathbb{E}(X)) = g(\mathbb{E}(X)) = \mathbb{E}(g(X)) \leq \mathbb{E}(f(X)) $
		\qed
	\end{proof}
	
	%
	%

	
	
	
	
	
	
	
\end{document}